\documentclass[12pt,leqeq]{article}
\usepackage{amssymb,amsthm,mathrsfs}
\usepackage{amsmath}
\usepackage{amscd}
\usepackage{xy}
\xyoption{all}
\usepackage[dvips]{graphicx}
\usepackage{graphpap}
\usepackage{pstricks}
\usepackage{pst-all}
\usepackage{pstcol}
\usepackage{color}
\newtheorem{thm}{Theorem}
\newtheorem{lem}[thm]{Lemma}
\newtheorem{cor}[thm]{Corollary}
\newtheorem{prop}[thm]{Proposition} 
\newtheorem{rem}[thm]{Remark}

\newtheorem{defn}[thm]{Definition}

\newcommand{\parcial}[2]{\frac{\partial#1}{\partial#2}}

  \def \e {{\epsilon}}

\newcommand{\normv}[1]{ \vert #1 \vert }

\date{}
\begin{document}
\setlength{\baselineskip}{16pt}
\title{Indirect \ Influences \ on \ Directed \ Manifolds}
\author{Leonardo Cano \ \ \  \mbox{and}  \ \ \ Rafael D\'\i az}

\maketitle

\begin{abstract}
We introduce a program aimed to studying problems arising from the theory of complex networks with differential geometric means.
We study the propagation of influences on manifolds assuming that at each point only a finite number of propagation velocities are allowed.
This leads to the computation of the volume of the moduli spaces of directed paths, i.e. paths satisfying the imposed tangential restrictions.
The proposed settings  provide a fertile ground for research with potential applications in geometry, mathematical physics, differential equations, and combinatorics.
We establish the general framework, develop its structural properties, and consider a few basic examples of relevance. The interaction
between differential geometry and complex networks is a new and promising field of study.
\end{abstract}

\section{Introduction}

Our aim in this work is to lay down the foundations for the study of the propagation of influences on directed manifolds. Our object of study can be approached from quite different viewpoints as indicated in the following, non-exhaustive, diagram:

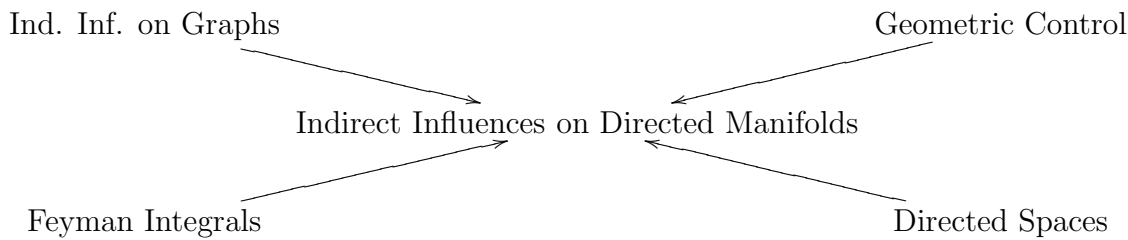
\[
\xymatrix @R=.3in  @C=.001in
{
\mbox{Ind. Inf. on Graphs} \ar[dr]& & \mbox{Geometric Control} \ar[ld]\\
&\mbox{Indirect Influences on Directed Manifolds}  & \\
 \mbox{Feyman Integrals} \ar[ur]& & \mbox{Directed Spaces} \ar[lu]
}
\]

\

Our departure point is the theory of indirect influences for weighted directed graphs which has gradually emerged thanks to the efforts of several authors, among them  Brin,  Chung, Estrada, Godet, Hatano, Katz, Page, Motwani, and  Winograd. Although the  history of the subject is yet to be written, for our purposes we may consider the introduction of the  Katz's index \cite{k} as an early modern approach to the problem of understanding the propagation of influences in complex networks. Fundamental developments in the field came with the introduction of the MICMAC \cite{godet}, PageRank \cite{b2,b1}, Communicability \cite{estrada}, and Heat Kernel \cite{chung} methods. In 2009 the second author proposed the PWP method for computing the propagation of influences on networks \cite{diaz}. In a nutshell the method proceeds as follows. We assume as  given a network (weighted directed graph) represented by its adjacency matrix $\ D, \ $also called the matrix of direct influences. Then one defines the matrix $\ T=T(D)\ $ of indirect influences  whose entry $\ T_{ij}\ $ measures the weight of the indirect influences exerted by vertex $\ j \ $ on vertex $\ i.\ $ The matrix $\ T \ $ is computed using the following expression:
$$T \  =  \ \frac{1}{e^{\lambda}-1}\sum_{n=1}^{\infty} D^n \frac{\lambda^n}{n!},$$
where $\ \lambda \ $ is a positive real parameter. In words: indirect influences arise from the concatenation of direct influences, and the weight of a concatenation of length $\ n \ $ comes from the product of $\ n \ $ entries of $\ D \ $ and the factor $\ \frac{\lambda^n}{n!}\ $ ensuring convergency by attaching a rapidly decreasing weight to longer chains of direct influences.
The PWP method has been applied to analyse educational programs, and to study indirect influences in international trade \cite{diazgomez}. Further extensions and applications  are underway. The stability of the method
with respect to changes in the matrix $\ D \ $ and the parameter $\ \lambda \ $ has been  recently  studied in  \cite{diazvargas}. \\

Our first proposal in this work is that one may regard a differential manifold provided with a tuple of vector fields on it -- we call such an object a directed manifold -- as being a smooth analogue of a directed graph with numbered outgoing edges attached to each vertex. Armed with this intuition we pose the question: Is there an extension of the theory of indirect influences from the discrete to the smooth settings? We argue that the answer is in the affirmative, and that such an extension both interplays with many notions already studied in the literature, e.g. control theory \cite{AgrachevSachkov, Jurdjevic, Rifford, Sontag, Zabczyk}, Feynman integrals \cite{SimonFunctional, glimjaffe, Zinn-Justin}, and directed topological spaces \cite{grandis}, and also demands the introduction of new ideas.\\

Constructing a smooth analogue for the PWP matrix of indirect influences -- whose entries are given by sums over paths in a graph -- leads
directly to Feynman path integrals, understood in the general sense of integrals over spaces of paths on manifolds. Although of great interest, we follow an
alternative approach in order to avoid the usual difficulties that have prevented, so far,  the development of a fully rigorous general integration theory over infinite dimensional manifolds. Thus, in order to reduce our computations to finite dimensional integrals, we impose strong tangential conditions on the allowed paths in our domains of integration.
\\

The background upon which we develop our constructions is the category of directed manifolds, introduced in Section \ref{bd}, which is also a convenient category for  studying geometric control theory. Our constructions bring about a new set of problems to geometric control theory -- usually  focus on the path reachability  and path optimization problems -- namely, the problem of computing integrals over the moduli spaces of directed paths. We remark again that strong tangency conditions are imposed in order to insure that the moduli spaces of directed paths -- also called the spaces of indirect influences -- split naturally into infinitely many finite dimensional pieces, each coming with a natural measure. Thus we have a notion of  integration over each piece, which we extend additively to the whole moduli space of directed paths, leaving the convergency of these sums to a case by case analysis. Fortunately, in our examples
we do obtain convergent sums. These ideas are developed in Section \ref{msi}, where we also introduce the wave of influences $\ u(p,t) \ $ which computes the total influence received by a point $\ p \ $ in time $\  t,$ i.e. $\ u(p,t) \ $ computes the volume of the moduli space of directed paths starting at an arbitrary point
and ending up at $\ p \ $ in time $\ t.$ \\

Our notion of directed manifolds is strongly related to the notion of directed spaces introduced by Grandis \cite{grandis}, and to some extend the former notion may be regarded as a smooth analogue of the latter. In Section \ref{msi} we make this connection precise.\\

In Section \ref{iilp} we discuss invariant properties for directed manifolds and for the moduli spaces of directed paths on them. We also study invariant properties with respect to  reordering of our given tuple of vectors fields. We propose a possible route for using our spaces of indirect influences to approach integrals with more general domains of integration, such as Feynman path integrals. Whether this approach can actually be implemented to work as a viable computational technique is left for future research. In Section \ref{iipqm} we study the moduli spaces of directed paths on product and quotient of directed  manifolds.\\

In Section \ref{iicvf} we study the moduli spaces of directed paths arising from constant vector fields on affine spaces. We show that even in this case, the simplest one, our theory yields results worth studying  where explicit computations are available. These settings give rise to fruitful constructions in combinatorics and
probability theory \cite{cd}. \\

Finally, in the closing Section \ref{qii} we indicate how our general settings for computing indirect influences, based on the computation of the volume of moduli spaces of directed paths,   can be extended to the quantum context adopting a Hamiltonian viewpoint. \\

\noindent \textbf{Notation}. For $\ n\in \mathbb{N}, \ $ we set $\ \ [n]=\{1,...,n\}, \ $ $ \ [0,n]  =  \{0,...,n\},  \ $ and
let $\ \ \mathrm{P}[n]\ \ $ be the set of subsets of $\ [n]. \ $
The amalgamated sum of closed subintervals of the real line $\ \mathbb{R}\ $ is given by $$ [a,b]\coprod_{b,c} [c,d] \ = \ [a,b+d-c]. $$ We let $\ \delta_{ab} \ $ be the
Kronecker's delta function.

\section{Basic Definitions}\label{bd}

We let $\ \mathrm{diman} \ $ be the category of directed manifolds.  A directed manifold is a tuple $\ (M, v_1,...,v_k)\ $  where $\ M \ $ is a smooth manifold, and $\ v_1,...,v_k \ $  are smooth vector fields on $M$, with $\ k\geq 1. \ $ A morphism $\ (f,\alpha) : (M, v_1,...,v_k) \longrightarrow (N, w_1,...,w_l)\ $ in $\ \mathrm{diman}\ $ is  a pair $\ (f, \alpha) \ $ where $\ f:M \longrightarrow N\ $ is a smooth map,  $\ \alpha:[k] \longrightarrow [l]\ $ is a map, and the following identity holds  $$df(v_i)\ = \ w_{\alpha(i)}, \ \ \ \ \ \mbox{for} \ \ \  i\in [k].$$
Let $\ (g, \beta): (N, w_1,...,w_l) \longrightarrow (K, z_1,...,z_r)\ $ be another morphism. The composition morphism $\ (g,\beta)\circ (f, \alpha)\ $ is given by:
$$(g,\beta)\circ (f, \alpha) \  =  \ (g \circ f,\beta \circ\alpha).$$ It satisfies the required property since
$$d(gf)(v_i)\  =  \ dg(df(v_i)) \  =  \ dg(w_{\alpha(i)})\  =  \ z_{\beta(\alpha(i))} \  =  \  z_{\beta\circ \alpha(i)}.$$

One can think of a directed manifold $\ (M, v_1,...,v_k)\ $ as being a smooth analogue of a "finite directed graph with up to $k$ outgoing numbered edges at each vertex".  Points in the manifold $\ M \ $ are thought as vertices in the smooth graph. The tangent  vectors $\ v_i(p) \in T_pM \ $ are thought as infinitesimal edges starting at $\ p. \ $ The out-degree of a vertex $\ p \in M\ $ is the number of non-zero infinitesimal edges starting at $\ p, \ $ i.e. the  cardinality of the set $\ \{ i \in [k] \ | \ v_i(p) \neq 0\}. $\\

An actual edge from $\ p \ $ to $\ q, \ $ points in $\ M, \ $ is a smooth path $\ \varphi: [0,t] \longrightarrow M\ $ with $\ \varphi(0)=p, \ $   $\ \varphi(t)=q, \ $ and such that the tangent vector at each point of $\ \varphi \ $ is an infinitesimal edge, i.e. $\ \dot{\varphi} =  v_i(\varphi) \ $ for some $\ i \in [k], \ $ or more explicitly
$$\dot{\varphi}(s) \ = \ v_i(\varphi(s)) \ \ \ \ \ \mbox{for \ all} \ \ \ s  \in [0,t].$$
We say that $\ p \ $ exerts a direct influence, in time $\ t>0, \ $  on the vertex $\ q \ $ through the path $\ \varphi. \ $ Note that $\ \varphi \ $ is determined by the initial point $\ p, \ $ and the index $\ i \ $ of vector field $\ v_i, \ $ thus we are entitled to use the notation $\ \varphi(t)=  \varphi_i(p,t),\ $ where $\ \varphi_i \ $ is the flow generated by the vector field $\ v_i$.\\

\begin{defn}{\em Let $\ (M, v_1,...,v_k)\ $ be a directed manifold and  $\ p, q \in M. \ $ The set of one-direction paths  $\ D_{p,q}(t)\ $ from $p$ to $q$ developed in time $\ t >0 \ $ is given by
$$D_{p,q}(t)\  =  \ \{ i \in [k] \ | \ \varphi_i(p,t)=q\}. $$
We also set
$$D_{p,q}(0) \  =  \
\left\{ \begin{array}{lcl}
\{p\} \ \ \ \   \mbox{if} \ \ p=q,\\
& & \\
\ \emptyset \ \ \ \ \   \mbox{ otherwise},
\end{array}
\right.$$
i.e. each point of $\ M\ $ exerts a direct influence over itself in time $\ t=0, \ $ and there are no $\ t=0 \ $ direct influences between different points of $\ M. \ $ Thus $\ D_{p,q}\ $ defines a map $\   D_{p,q}:\mathbb{R}_{\geq 0} \ \longrightarrow \ \mathrm{P}[k].$ We also say that $\ D_{p,q}(t)\ $ is the set of direct influences
 from $\ p\ $ to $\ q \ $ exerted in time $\ t>0. \ $
}
\end{defn}

There might also be one-direction paths from $\ p \ $ to $\ q \ $ taking an infinite long interval of time to be exerted, these influences occur through a path $\ \varphi: \mathbb{R} \longrightarrow M \ $ such that  $\ \underset{t \rightarrow -\infty}{\mathrm{lim}}\varphi(t) =  p\ $ and $\ \underset{t \rightarrow \infty}{\mathrm{lim}}\varphi(t) =  q. \ $ Semi-infinite direct influences can be similarly defined.
One might also consider topological direct influences from $\ p \ $ to $\ q \ $ which are exerted through a path $\ \varphi: \mathbb{R} \longrightarrow M\ $ such that  $\  p \in \underset{t \rightarrow -\infty}{\omega\mathrm{lim}}\varphi(t) \  $ and $\  q \in \underset{t \rightarrow \infty}{\omega\mathrm{lim}}\varphi(t). \ $
We will no further consider one-direction paths of these types in this work. \\

Next we introduce the notion of indirect influences which arise from the concatenation of direct influences. Our focus is on finding a convenient parametrization for the space of all such concatenations.

 \begin{defn}\label{ii}{\em Let $\ (M, v_1,...,v_k)\ $ be a directed manifold and $\ p,q \in M. \ $
 A directed path from $\ p \ $ to $\ q \ $  displayed in time $\ t>0\ $ through  $\ n\geq 0\ $ changes of directions  is  given by a pair $\ (c,s)\ $ with the following properties:

\begin{itemize}

\item $c=(c_0, c_1,...,c_n)\ $ is a $\ (n+1)$-tuple with $\ c_i \in [k] \ $ and such that $\ c_i \neq c_{i+1}$.  We say that $\ c \ $ defines the pattern (of directions) of the directed path $\ (c,s). \ $ We let $\ D(n,k) \ $ be the set of all such tuples, and $\ l(c)=n+1 \ $ be the length of $\ c. \ $ There are $\ k(k-1)^{n}\ $ different patterns in $\ D(n,k). \ $ Note that we may regard a pattern $\ c \ $ as a map $\ c: [0,n] \longrightarrow [k]. \ $

\item $s=(s_0,...,s_n)\ $  is a $\ (n+1)$-tuple with $\ s_i \in \mathbb{R}_{\geq 0\mathbb{}}\ $ and such that $\ s_0+\cdots+s_n=t. \ $ We say that $\ s \ $ defines the time distribution of the directed path $\ (c,s), \ $ and let $\ \Delta_{n}^t \ $ be the $n$-simplex of all such tuples.

\item The pair $\ (c,s)\ $ determines a $\ (n+2)$-tuple of points $ \ (p_0,\ldots,p_{n+1})\ \in \ M^{n+2} \ $
given by: $$\ p_0\ = \ p \ \ \ \ \ \mbox{and} \ \ \ \ \ p_{i}\ = \ \varphi_{c_{i-1}}(p_{i-1},s_{i-1})  \ \ \ \ \mbox{for}\ \ \ \ 1 \leq i \leq n+1,$$  where
$\ \varphi_{c_{i-1}} \ $ is the flow generated by the vector field $\ v_{c_{i-1}}. \ $
We denote the last point $\ p_{n+1} \ $ by $\ \varphi_c(p,s). \ $

\item The pair $\ (c,s) \ $ must be such that $\ \ \varphi_c(p,s)=  q.$

\item Directed paths in time $\ t=0 \ $ are the same as one-direction paths in time $\ t=0. \ $
\end{itemize}
}
\end{defn}

\begin{rem}{\em By definition directed paths include one-direction paths as well, even for the conventional case $\ t=0. \ $
We also say that $\ (c,s) \ $ determines  an indirect influence from $\ p \ $ to $\ q \ $ exerted in time $\ t. \ $ The fact that our paths are displayed in non-negative time means that indirect influences propagate forward in time.
}
\end{rem}

\begin{rem}\label{gmii}
{\em The geometric meaning of directed paths is made clear through the following construction.
A pair $\ (c,s)\ $ as above determines a piece-wise smooth path $$\varphi_{c,s}:[0, s_0 + \cdots + s_n] \ \ \simeq \ \ [0,s_0] \ \underset{s_0,0}{\bigsqcup}\ \cdots \ \underset{s_{n-1},0}{\bigsqcup} \ [0,s_n] \ \ \longrightarrow \ \ M  $$ such that the restriction of $\ \varphi_{c,s}\ $ to the interval $\ [0,s_i], \ $ for $0 \leq i \leq n, \ $ is given by
$$\varphi_{c,s}|_{[0,s_i]}(r) \  =  \ \varphi_{c_i}(p_{i}, r ) \ \ \ \ \ \mbox{for \ all} \ \ \ r\in [0,s_i].$$
We say that $\ \varphi_{c,s} \ $ is the directed path determined by the pair $\ (c,s). \ $ Indirect influences are exerted through such directed paths. Whenever necessary we write $\ \varphi_{v,c,s} \ $ instead of $\ \varphi_{c,s}\ $ to make explicit that these paths do depend on the vector fields $\ v=(v_1,...,v_k).\ $ Figure 1 shows the directed path associated to a pair
$\ (c,s)$.}
\end{rem}

\

\begin{center}
\psset{unit=0.5cm}
\begin{pspicture}(-20,1)(2,9)
    \pscurve[](-19,2)(-20,3)(-16,8)(-1,6)(0,4)(-1,2)
    (-5,3)
    (-10,2)(-16,2)(-17,2.5)(-18,2.5)(-18.5,2)(-19,2)
    \pscurve[showpoints=true](-18,4)(-14,6)(-11,5)(-7,7)(-3,5.7)
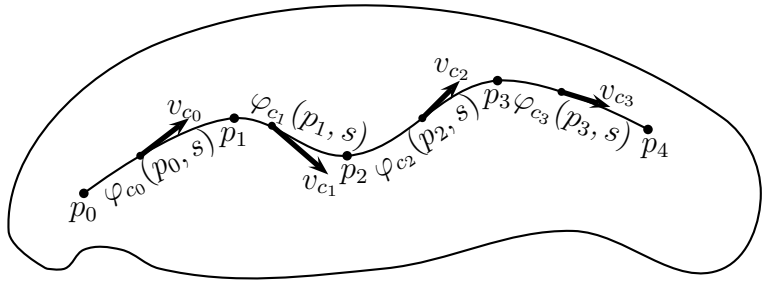
    \rput(-10,-0.5){Figure 1. Directed path associated to a pair $\ (c,s)$.}
    \rput(-18,3.5){$p_0$}
    \psline[linewidth=2pt]{->}(-16.5,5)(-15.2,6)
    \psdot[dotstyle=*,
    dotsize=3pt](-16.5,5) 
    \rput(-15.3,6.2){$v_{c_0}$} 
    \rput{28}(-16,4.7){$\varphi_{c_0}(p_0,s)$}
    \psdot[dotstyle=*,
    dotsize=3pt](-13,5.8)
    \rput(-11.7,4.2){$v_{c_1}$} 
    \psline[linewidth=2pt]{->}(-13,5.8)(-11.5,4.5)
    \rput{-20}(-12,6){$\varphi_{c_1}(p_1,s)$}
    \rput(-14,5.5){$p_1$}
    \psdot[dotstyle=*,
    dotsize=3pt](-9,6)
    \psline[linewidth=2pt]{->}(-9,6)(-8,7)
     \rput(-8.2,7.3){$v_{c_2}$}
     \rput{35}(-8.9,5.5){$\varphi_{c_2}(p_2,s)$}
     \rput(-10.8,4.5){$p_2$}
     \rput(-7,6.5){$p_3$}
     \psdot[dotstyle=*,
    dotsize=3pt](-5.3,6.7)
    \psline[linewidth=2pt]{->}(-5.3,6.7)(-4,6.3)
    \rput(-3.8,6.6){$v_{c_3}$}
    \rput{-20}(-5,6){$\varphi_{c_3}(p_3,s)$}
    \rput(-2.8,5.2){$p_4$}
    \end{pspicture}
\end{center}

\

\

Note that directed paths in the sense above are examples of horizontal paths as defined in geometric control theory \cite{AgrachevSachkov}.

\begin{rem}{\em Our notion of indirect influences on directed manifolds may be regarded as a limit case of the propagation of disturbances in geometric optics, see Arnold \cite{arnold}. In geometric optics one works with a Riemannian manifold $\ M, \ $ and is given a map $\ v:SM \longrightarrow \mathbb{R}_{\geq 0}\ $ from the unit sphere bundle of $\ M \ $ to the non-negative real numbers. The number $\ v(l)\ $ gives the speed allowed for the propagation of a disturbance  along the direction $\ l. \ $ Indirect influences on a directed manifold $\ (M,v_1,...,v_k)\ $ correspond to the propagation of disturbances in geometric optics, if one lets $\ v \ $ be the singular map that is zero everywhere except at the directions defined by the vector fields $\ v_j, \ $ and on this directions it assumes the values $\ |v_j|. \ $ Note  that the notion of indirect influences does not demand a Riemannian structure on $M$. Figure 2 illustrates the relation between indirect influences and geometric optics, by displaying the deformation of the indicatrix surface
(the image of $\ v$) from a smooth ellipses to a curve concentrated on three vectors.

}
\end{rem}

\

\

\begin{center}
\psset{unit=0.5cm}
\begin{pspicture}(-20,1)(2,7)
        \psline[linewidth=1pt]{->}(-2,5)(-2,7)
        \psline[linewidth=1pt]{->}(-2,5)(1,5)
        \psline[linewidth=1pt]{->}(-2,5)(-3.5,4)
        \pscurve[](-2.3,6)(-2,7.2)(-1.8,7.2)
        (-1.8,5.8)(-1.8,5.2)(-1.2, 5.2)(0.8,5.2)(1.2,5.2)(1.2,4.8)
        (1,4.8) (0.8,4.7)(-1.2, 4.8)
        (-1.8,4.8)(-2,4.8)(-3,3.9)(-3.7,3.9)
        (-3.7,4.1)(-3,4.7)(-2.2,5.1)(-2.3,6)
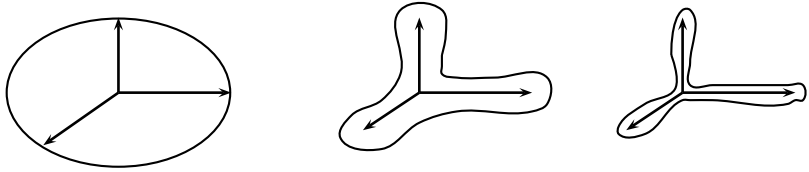
            \rput(-10,0.5){Figure 2. Indirect influences as a limit case of propagations in geometric optics.}
    \psellipse(-17,5)(3,2)              
    \psline[linewidth=1pt]{->}(-17,5)(-17,7)
    \psline[linewidth=1pt]{->}(-17,5)(-14,5)
    \psline[linewidth=1pt]{->}(-17,5)(-19,3.6)
    \psline[linewidth=1pt]{->}(-9,5)(-6,5)
    \psline[linewidth=1pt]{->}(-9,5)(-9,7)
    \psline[linewidth=1pt]{->}(-9,5)(-10.5,4)
    \pscurve[](-8.4,6)(-8.3,6.5)
    (-8.4,7.2)(-9.5,7.2)(-9.5,5.5)
    (-10,4.8)(-10.8,4.4)(-11.1,3.8)(-10,3.5)(-9,4.2)
    (-8,4.5)(-6,4.5)(-5.6,4.7)(-5.6,5.4)(-7,5.4)(-8,5.4)(-8.4,5.5)(-8.4,5.7)(-8.4,6)
\end{pspicture}
\end{center}

\

\begin{rem}\label{ifm}{\em Although not strictly necessary, for simplicity we usually assume   that the flows generated by the vector fields $\ v_j \ $ are globally defined by smooth maps $$\ \varphi_j(\ , \ ): M \times \mathbb{R} \ \longrightarrow  \ M \ $$ yielding a one-parameter group of diffeomorphisms of $M$:
\begin{itemize}
\item The map $\ \varphi_j( \ , s): M \longrightarrow M\ $ is a diffeomorphism for all $\ s \in \mathbb{R}$.
\item The group condition  $\  \varphi_j(\varphi_j(p,s_1), s_2)\ =  \ \varphi_j(p, s_1+s_2) \ \ \ \mbox{holds \ for} \ \ s_1, s_2 \in \mathbb{R}.$
\end{itemize}
A pattern $\ c \in D(n,k) \ $ defines an iterated flow given by the smooth  map
$$\varphi_c: M \times \mathbb{R}^{n+1}\  \longrightarrow \ M$$ defined by recursion on the length of $\ c \ $ as follows:
$$\varphi_c(p, s_0,...,s_n) \   = \  \varphi_{c_n}(\varphi_{c|_{[0, n-1]}}(p, s_0,...,s_{n-1}), s_n).$$
Fixing a time distribution $\ (s_0,...,s_n)\ $ we obtain the diffeomorphism $$\varphi_c( \ \ , s_0,...,s_n): \ M \ \longrightarrow \ M .$$
These construction justify the notation $\ \varphi_c(p,s)\ $ for the point $\ p_{n+1}(c,s)\ $ introduced in Definition \ref{ii}.\\
}
\end{rem}

We regard the $n$-simplex $\ \Delta_{n}^t\ $ introduced in Definition \ref{ii} as a smooth manifold with corners. There are at least three different approaches to differential geometry on manifolds with corners. First we can apply differential geometric notions on the interior of $\ \Delta_{n}^t. \ $ Also it is possible to introduce differential geometric objects on  $\ \Delta_{n}^t \ $ by considering objects that are smooth on an open neighborhood of $\ \Delta_{n}^t\ \ $ in
$\ \ \mathbb{R}^{n+1}. \ $ A third and more intrinsic approach for doing differential geometry on  $\ \Delta_{n}^t \ $ relies on deeper results in the theory of manifolds with corners. For a fresh approach the  reader may consult \cite{joyce}. Although this more comprehensive approach is certainly desirable, for simplicity, we will not further consider it.
\begin{prop}{\em
For a pattern $\ c \in D(n,k), \ $ the map
$$\varphi_c: M \times \Delta_{n}^t \ \longrightarrow \ M$$ sending a pair $\ (p,s) \in M \times \Delta_{n}^t \ $ to the point $\ \varphi_c(p,s) \in M \ $ is a smooth map and a diffeomorphism for a fixed time distribution $\ s \in \Delta_{n}^t.$ }
\end{prop}

Next we introduce the main objects of study in this work, namely, the moduli spaces of directed path, also called the spaces of indirect influences, on directed manifolds. These spaces parametrize directed paths
from a given point to another.

\begin{defn}\label{ii}
{\em Let $\ (M, v_1,...,v_k) \ $ be a directed manifold and $\ p,q \in M. \ $ The moduli  space  $\ \Gamma_{p,q}(t)\ $ of directed paths  from $\ p \ $ to
$\ q \ $ developed in time $\ t>0 \ $ is given by
$$\Gamma_{p,q}(t)\  = \ \Big\{ (c,s) \ \Big| \ \varphi_c(p,s)=q\Big\}  \ = $$
 $$ \coprod_{n=0}^{\infty}\coprod_{c \in D(n,k)} \{ s \in \Delta_{n}^{t} \ | \ \varphi_c(p,s)=q\} \  \ = \ \
\coprod_{n=0}^{\infty}\coprod_{c \in D(n,k)} \Gamma_{p,q}^c(t). $$ In addition we set
$$\Gamma_{p,q}(0) \  =  \ \Gamma_{p,q}^{\emptyset}(0) \  =  \
\left\{ \begin{array}{lcl}
\{p\} \ \ \ \   \mbox{if} \ \ p=q,\\
& & \\
\ \emptyset \ \ \ \ \   \mbox{ otherwise},
\end{array}
\right.$$
Figure 3 \ shows a schematic picture of a component $\ \Gamma_{p,q}^c(t) \ $ of the moduli space of indirect influences.
}
\end{defn}

\begin{center}
\psset{unit=0.5cm}
\begin{pspicture}(-20,1)(2,9)
        \psline[linewidth=1pt]{}(-12,3)(-9,7)
        \psline[linewidth=1pt]{}(-9,7)(-5.5,3.5)
        \psline[linestyle=dashed,dash=3pt 2pt]{}(-12,3)(-5.5,3.5)
        \psline[linewidth=1pt]{}(-12,3)(-8,2)
        \psline[linewidth=1pt]{}(-8,2)(-5.5,3.5)
        \psccurve[fillstyle=solid,fillcolor=lightgray](-11,4.2)(-9.5,4.3)(-8.5,3.5)(-7,5)(-9,5)
        \psline[linewidth=1pt]{}(-9,7)(-8,2)
        \pscurve[linewidth=1pt]{->}(-8,4.5)(-6,6)(-5,6)
        \rput(-3,6){$\varphi_{c}(p,s)=q$}
         \rput(-10,0.5){Figure 3. Moduli space of directed paths $\ \Gamma_{p,q}^c(t). \ $ }
    \end{pspicture}
\end{center}

\

\begin{rem}{\em For a fixed pattern $\ c \ $ the continuity of the iterated flow $\ \varphi_c(p, \ )\ $ implies that the moduli space of directed paths $\ \Gamma_{p,q}^c(t)\ $ is compact, as it is a  closed subspace of $\ \Delta_{n}^{t}.$
}
\end{rem}

The moduli spaces of directed paths  come equipped  with the structure of a category. Indeed directed paths are pretty close of being the free category generated by one-direction path, but not quite since we have ruled out repeated directions.

\begin{thm}{\em Altogether the moduli spaces of directed paths on a directed manifold  form a topological category.
}
\end{thm}

\begin{proof}Given a directed manifold $\ (M, v_1,...,v_k)\ $ we let $\ \Gamma =  \Gamma(M, v_1,...,v_k)\ $ be the category of directed paths on $\ M. \ $ The objects of $\ \Gamma \ $  are the points of $\ M. \ $ Given $\ p,q \in M, \ $ the  space of morphisms in  $\ \Gamma \ $ from $\ p \ $ to $\ q \ $ is given by
$$\Gamma_{p,q} \  =  \ \coprod_{n \in \mathbb{N}}\coprod_{c\in D(n,k)}\Gamma_{p,q}^c
\ \ \ \ \ \mbox{where} \ \ \ \ \ \Gamma_{p,q}^c \ = \ \{(s,t) \in \mathbb{R}_{\geq 0}^{n+2} \ \ | \ \ s \in \Delta_n^t,  \ \ \varphi_{c}(p,s)=q\}.
$$ In order to define continuous composition  maps $ \ \ \circ: \Gamma_{p,q}\ \times \ \Gamma_{q,r} \ \longrightarrow \ \Gamma_{p,r} , \ \ $ it is enough to define componentwise
composition maps $$\circ:\Gamma_{p,q}^c \ \times \ \Gamma_{q,r}^d \ \longrightarrow  \ \Gamma_{p,r}^{c\ast d}$$
for given patterns $\ c \ $ and $\ d \ $ with $\ n=l(c)\ $ and $\ m=l(d). \ \ $ We  consider two cases: \\

\noindent $\bullet$ If $\ c_{n} \neq d_0, \ $ then $\ c\ast d \ = \ (c,d)\ $ and
$$(s_0,...,s_n) \circ (u_0,...,u_m) \   =  \ (s_0,...,s_n, u_0,...,u_m).$$

\noindent $\bullet$ If $\ c_{n} = d_0, \ $ then $\ c\ast d \  =  \ (c_0,..., c_n) \ast (d_0,...,d_m) \  =  \ (c_0,..., c_n,d_1,...,d_m)\ \ $ and
$$(s_0,...,s_n) \circ (u_0,...,u_m) \  =   \ (s_0,...,s_n+ u_0,...,u_m).$$
These compositions are well-defined continuous maps satisfying the associative property.
The unique $\ t=0\ $ directed path from  $\ p \in M \ $ to itself gives the identity morphism  for each object $\ p \in \Gamma.$
\end{proof}

\begin{rem}{\em The moduli spaces of directed paths $\ \Gamma_{p,q}(t) \ $ can be extended from points to arbitrary subsets of $\ M \ $ as follows. Given $\ A,B \subseteq M \ $ we define the moduli space of directed paths from $\ A \ $ to $\ B  \ $ as
$$\Gamma_{A,B}(t)\  =  \ \Big\{ (c,s) \ \Big| \ p\in A, \ \ \varphi_{c}(p,s) \in B \Big\} \  = $$  $$ \coprod_{n=0}^{\infty}\coprod_{c \in D(n,k)} \{ s \in \Delta_{n}^{t} \ | \ p \in A, \ \varphi_c(p,s)\in B\} \  =  \
\coprod_{n=0}^{\infty}\coprod_{c \in D(n,k)} \Gamma_{A,B}^c(t). $$
Restricting attention to embedded oriented submanifolds of $M$, and following techniques from Chas and Sullivan's string topology \cite{chas-sullivan}, this construction gives rise to some kind of transversal category.
}
\end{rem}

We close this section introducing a few subsets of $\ M \ $  useful for understanding the propagation of influences on $\ M. \ $ These sets are usually called the reachable sets in geometric control theory, and are natural generalizations of the corresponding graph theoretical notions.
They also play a prominent role in general relativity \cite{penrose}.  For $\ A \subseteq M \ $ we set:

\begin{itemize}
\item $\Gamma_A(t)  =  \{q \in M \ | \ \Gamma_{A,q}(t) \neq \emptyset \}\ $  is the set of points in $\ M \ $ influenced by $\ A \ $ in time $\ t$.
\item $\Gamma_{A,\leq}(t)  =  \{q \in M \ | \ \mbox{there is} \ \ 0 \leq s \leq t, \ \ \mbox{such that} \ \ \Gamma_{A,q}(s) \neq \emptyset \}\ $ is the set of points in $\ M \ $ influenced by $\ A \ $ in time less or equal to $\ t$.
\item $\Gamma_A =  \{q \in M \ | \ \Gamma_{A,q}(t) \neq \emptyset \ \  \mbox{for some} \ \  t \geq 0\}\ $ is the set of points in $\ M \ $ that are influenced by $\ A. $
\item $\Gamma_A^-(t) =  \{q \in M \ | \ \Gamma_{q,A}(t) \neq \emptyset \}\ $ is the set of points in $\ M \ $ that influence $\ A \ $ in time $\ t, \ $ i.e.  the set of points on which $A$ depends on time $\ t.$
\item $\Gamma_{A, \leq}^{-}(t)  = \{q \in M \ | \ \mbox{there is} \ \ 0 \leq s \leq t, \  \mbox{such that} \ \Gamma_{q,A}(s) \neq \emptyset \}\ $ is the set of points in $\ M \ $ that influence $\ A \ $ in time less or equal to $\ t.$
\item $\Gamma_A^-  =  \{q \in M \ | \ \Gamma_{q,A}(t) \neq \emptyset \ \  \mbox{for some} \ \  t \geq 0\}\ $ is the set of points in $\ M \ $ that influence $\ A. $
\item $\mathrm{F}_A(t)= \partial \Gamma_{A,\leq}(t)\ \ $  and $\ \ \mathrm{F}_A^-(t)=\partial \Gamma_{A, \leq}^{-}(t) $ are called, respectively, the front of influence and the front of dependence of $\ A \ $ in time $\ t$.
\end{itemize}

Note that a directed manifold $\ M \ $ is naturally a pre-poset by setting
$$p\leq q \ \ \ \ \ \mbox{if \ and \ only \ if} \ \ \ \ \  q \in  \Gamma_p.$$ The associated poset is the quotient space
$\ M_{\sim},\ $ where the equivalence relation $\ \sim\ $ on $\ M \ $ is given by
$$p \sim q \ \ \ \ \ \mbox{if \ and \ only \ if} \ \ \ \ \ q\in \Gamma_{p} \ \ \ \mbox{and} \ \ \ p \in \Gamma_q.$$ The space $\ M_{\sim}\ $ tell us how $\ M \ $ splits into components of co-influences, i.e. the path connected components of $\ M \ $ through directed paths.\\

Note that a directed manifold $\ (M,v_1,...,v_k)\ $ comes equipped with a natural distribution, indeed for each point $\ p \in M \
 $ we have the subspace
$$<v_1(p),...,v_k(p)> \ \ \subseteq \ \ T_pM$$
generated by the vectors $\ v_1(p),...,v_k(p).\ $ If this distribution is integrable, then directed paths are confined to live on the leaves.
Thus to study the moduli spaces of directed paths, in the integrable case, we may as well forget about the manifold $\ M \ $ and work leaf by leaf. So the interesting
cases of study are:
\begin{itemize}
\item $<v_1(p),...,v_k(p)> \  =  \ T_pM, \ $ i.e. $\ M \ $ itself is the unique leaf.
\item The distribution $\ <v_1(p),...,v_k(p)> \  \subseteq  \ T_pM\ $ is not integrable.
\end{itemize}

\section{Measuring the Moduli Spaces of Directed Paths}\label{msi}

Fix a directed manifold $\ (M, v_1,...,v_k). \ $  In order to measure directed paths on $\ M \ $  we assume from now on  that an orientation on $\ M \ $ has been chosen. To gauge the amount of indirect influences exerted, in time $t$, by a point $\ p\in M \ $ on a point $\ q \in M\ $  we need to define measures on the modulis spaces $\ \Gamma_{p,q}(t)\ $ of directed paths. From Definition \ref{ii}
we see that $\ \Gamma_{p,q}(t) \ $ is a disjoint union of pieces, one for each pattern $\ c \in D(n,k), \ $ of the form
$$\Gamma_{p,q}^c(t)\ =  \ \{ s \in \Delta_{n}^{t} \ \  | \  \ \varphi_c(p,s)=q\}.$$
So, our problem reduces to imposing measures on the pieces $\ \Gamma_{p,q}^c(t)$.\\

The $n$-simplex $\ \Delta_{n}^{t}\ $ is a smooth manifold with corners, and  comes equipped with a Riemannian metric and its associated volume form. Indeed using Cartesian coordinates $$l_1=s_0,\ \  \ l_2=s_0+s_1,\ \ \  ....... \ \ \ , \ \ \  l_{n}=s_0+ \cdots + s_{n-1}$$
the $n$-simplex can be identified with the following subset of $\ \mathbb{R}^{n}:$
$$\Delta_{n}^{t} \  =  \ \Big\{(l_1,...,l_{n}) \in \mathbb{R}^{n} \ \ | \ \  0\leq l_1 \leq l_2 \leq ...... \leq l_{n} \leq t\Big\}.$$
Thus $\ \Delta_{n}^{t} \ $ inherits a Riemaniann metric, an orientation, and the corresponding volume form   $\ dl_1 \wedge\cdots \wedge dl_{n}.$ With this measure we have that
$$\mathrm{vol}(\Delta_{n}^{t}) \  =  \ \frac{t^{n}}{n!}  \ \ \ \ \ \ \mbox{for} \ \ \ \ \ \ n \geq 0.$$

\begin{defn}{\em A directed manifold $\ (M, v_1,...,v_k)\ $ has smooth spaces of directed path if for any pattern $\ c \in D(n,k)\ $ and points $\ p,q \in M\ $ the space of indirect influences $\ \Gamma_{p,q}^c(t) \ $ is a smooth embedded sub-manifold of $\ \Delta_{n}^{t}. \ $ }
\end{defn}

For our next result we  use the implicit function theorem for manifolds  \cite{Guillemin, Warner}. Let $\ f:N \longrightarrow M\ $ be a smooth map between differential manifolds and fix $\ q \in M. \ $ Then $\ f^{-1}(q)\ $
is a smooth sub-manifold of $\ N \ $ if for each $ \ p\in f^{-1}(q)\ $ the linear map $$ d_pf: T_pN \longrightarrow T_{q}M $$ has maximal rank, that is
$$\mathrm{rank}(d_pf)\  =  \ \mathrm{min}\{{\mathrm{dim}(N) , \mathrm{dim}(M)}\}.$$
If $\ \mathrm{rank}(d_pf) =   \mathrm{dim}(N), \ $ then $\ d_pf \ $ is injective,  $\ f \ $ is an immersion, and $\ f^{-1}(q) \ $ is a set of isolated points. If $\ \mathrm{rank}(d_pf) =  \mathrm{dim}(M),\ $ then $\ d_pf \ $ is surjective, $\ f \ $ is a submersion, and $\ f^{-1}(q)\ $ is a sub-manifold of $\ N \ $ of dimension
$\ \ \mathrm{dim}(M) - \mathrm{dim}(N).$ \\

Next we apply this result to the open part of manifolds with corners.

\begin{thm}\label{sii}{\em  Let $\ (M, v_1,...,v_k)\ $ be a directed manifold.  Fix a pattern $\ c \in D(n,k)\ $ with $\ n\geq 1,\ $ and a point $\ p \in M.\ $ If for any $\ (s_0,...,s_n)\ $ in the open part of  $\ \Gamma_{p,q}(t)\ $ there are $\ \mathrm{min}(n,\mathrm{dim}(M))\ $ linearly independent vectors among the  vectors given for $\ i \in [0,n-1] \ $ by
$$d^M\varphi_{(c_{i+1},...,c_n)}\Big|_{(s_{i+1},...,s_n)}\big[v_{c_i}(\varphi_{c_0,..., c_i}(s_0,...,s_i))\big] \  -  \ v_{c_n}(\varphi_c(s_0,...,s_n))
\ \ \in \ T_{\varphi_{c}(p,s_0,...,s_n)},$$
 then $\ \Gamma_{p,q}^c(t) \ $ is a smooth sub-manifold of $\ \Delta_{n}^{t}$.
}
\end{thm}

\begin{proof}

Fix $\ c \in D(n,k)\ $ with $\ n\geq 1. \ $ Recall from Remark \ref{ifm} that $\ \varphi_c: M \times \mathbb{R}^{n+1} \longrightarrow M\ $ is the iterated flow associated to $\ c. \ $  The differential of $\ \varphi_c \ $ naturally split as:
$$d\varphi_c \ \ = \ \ d^M\varphi_c \ \ + \ \ d^{\mathbb{R}^{n+1}}\varphi_c .$$
Consider the map $\ \ \phi: \Delta_{n}^{t} \longrightarrow M \ \  \mbox{given by}$
$$\phi(s) \ = \ \phi(s_0,...,s_{n-1})\ = \ \varphi_c(p,s_0,...,s_{n-1},t-s_0- \cdots -s_{n-1}),$$ where we are using the identification
$$\Delta_{n}^{t}\  = \ \bigg\{\ s= (s_0,...,s_{n-1}) \in \mathbb{R}_{\geq 0} \ \ \bigg| \ \ |s|=s_0 + \cdots + s_{n-1} \leq t \ \bigg\}.$$
In order to guarantee that $\ \Gamma_{p,q}^c(t)\ =\  \phi^{-1}(p) \ $ is a smooth sub-manifold of $\ \Delta_{n}^{t}\ $ we impose the condition that $\ d_s \phi \ $ has maximal rank
for $\ s \in \phi^{-1}(p).\ $ Next we compute for $\  i \in [0,n-1]\ $ the vectors
$$ \frac{\partial \phi}{\partial s_i}(s) \ =  \ d_s \phi(\frac{\partial}{\partial s_i}) \ \ \in \ \ T_{\phi(s)} M.$$ Using the identity
$$\parcial{}{s_n}(\varphi_{c_0, \cdots, c_n})(p,s_0, \cdots, s_n)\  =  \ v_{c_n}(\varphi_{c_1, \cdots, c_n}(p,s_0, \cdots, s_n)),$$
one can  show that
$\ \frac{\partial \phi}{\partial s_i}(s)\ $ is given by
$$d^M\varphi_{c_{i+1},...,c_n}\Big|_{(s_{i+1},...,s_{n-1},s_n)}\big[v_{c_i}(\varphi_{c_0,..., c_i}(s_0,...,s_i))\big] \  -
\ v_{c_n}(\varphi_{c_0,..., c_n}(s_0,...,s_n)),$$
where we recall that $\ s_n\ = \ t-|s|,\ $
$$d^M\varphi_{c_{i+1},...,c_n} \ \ = \ \  d\varphi_{c_n}(\ \ , t-|s|)\ \circ \ \cdots \ \circ \ d^M\varphi_{c_{i+1}}(\ \ ,s_{i+1} ), \ \ \ \ \mbox{and}$$
$$\varphi_{c_0,..., c_i}(s_0,...,s_i) \ \ = \ \  \varphi_{c_{i}}[\varphi_{c_0,..., c_{i-1}}(s_0,...,s_{i-1}),s_i]\ \ \ \  \mbox{for}\ \ \ \ \ i \geq 1 .$$

\noindent Thus the rank of $\ d_s \phi \ $ is maximal at each point $\ s \in \phi^{-1}(q)\ $ if and only if there are exactly $\ \mathrm{min}(n,\mathrm{dim}(M))\ $ linearly independent vectors among the vectors $\ \frac{\partial \phi}{\partial s_i}(s)\ $ given by the expression above. We have shown the desired result.
\end{proof}

\begin{cor}\label{oc}
{\em Under the hypothesis of Theorem \ref{sii}, the interior of the moduli space $\ \Gamma_{p,q}^c(t)\ $ is an oriented Riemannian sub-manifold of $\ \Delta_n^t$.}
\end{cor}

\begin{proof} We use oriented differential intersection theory as developed by Guillemin \cite{Guillemin}. Since $\ \Gamma_{p,q}^c(t)\ $ is a smooth sub-manifold of $\ \Delta_n^t \ $ it acquires by restriction a Riemannian metric. The orientation on $\ \Gamma_{p,q}^c(t)\ $ arises as follows. For $\ s \in \Gamma_{p,q}^c(t)\ $ write $$T_s\Delta_n^t \ \ \simeq \ \  N_s\Gamma_{p,q}^c(t) \ \oplus \ T_s\Gamma_{p,q}^c(t),$$ where $\ N_s\Gamma_{p,q}^c(t) \   \simeq   \ T_s\Delta_n^t /T_s\Gamma_{p,q}^c(t) \ $ is the normal bundle of $\ \Gamma_{p,q}^c(t).\ $
Note that $$d_s\phi(T_s\Delta_n^t)\ = \ T_{\phi(s)}M \ \ \  \mbox{and \ thus}  \ \ \ d_s\phi: N_s\Gamma_{p,q}^c(t) \ \longrightarrow \ T_{\phi(s)}M \ \ \ \mbox{is \ an \ isomorphism}.$$ Since $\ T_s\Delta_n^t\ $ is oriented, and $\ N_s\Gamma_{p,q}^c(t)\ $ acquires an orientation from the isomorphism above, then
$\ T_s\Gamma_{p,q}^c(t)\ $ naturally acquires an orientation.
\end{proof}

For a directed manifold with a smooth moduli space of directed paths each piece $\ \Gamma_{p,q}^c(t)\ \subseteq \ \Delta_{n}^{t}\ $ acquires from $\ \Delta_{n}^{t}\ $ a Riemannian metric. If in addition we assume that each piece $\ \Gamma_{p,q}^c(t)\ $ is given an orientation, then $\ \Gamma_{p,q}^c(t)\ $ acquires a volume form denoted by $\ dl_c. \ $ As we have just shown this is the situation arising from the conditions of Theorem \ref{sii}. \\

 We are ready to highlight a few functions on the moduli spaces of directed paths, for a fix a time $\ t>0, \ $ that one would like to integrate against these measures. \\

\noindent \textbf{1. Volume of Moduli Space of Directed Paths.} \\

Each component $\ \Gamma_{p,q}^c(t)\ $ of the space of indirect influences is compact and thus of bounded volume. We  define the volume or total measure  of $\ \Gamma_{p,q}(t), \ $ leaving convergency issues to be discussed on a case by case basis, as follows:
$$\mathrm{vol}(\Gamma_{p,q}(t))\  =  \ \int_{\Gamma_{p,q}(t)}1dl \  =  \ \sum_{n=1}^{\infty}\ \sum_{c\in D(n,k)}\ \int_{\Gamma_{p,q}^c(t)}1\ dl_c \  =  \ \sum_{n=1}^{\infty}\ \sum_{c\in C(n,k)}\mathrm{vol}(\Gamma_{p,q}^c(t)) .$$

\noindent \textbf{2. Functions on directed paths coming from differential $1$-forms on $M$.} \\

Let $A$ be a differential $1$-form on $M$.  We formally write
$$\int_{\Gamma_{p,q}(t)}\widehat{A}\ dl \ \ = \ \  \sum_{n=1}^{\infty}\ \sum_{c\in D(n,k)}\ \int_{\Gamma_{p,q}^c(t)}\widehat{A} \ dl_c,$$
where the map $\ \widehat{A}:\Gamma_{p,q}^c(t) \ \longrightarrow \ \mathbb{R} \ $ is given by
$$A(c,s) \ \ = \ \ \int_0^t \varphi_{c,s}^{*}A \ \ = \ \  \sum_{i=0}^{l(c)}\ \int_{0}^{s_i}\varphi_{c,s}|_{[0,s_i]}^{\ast}A,$$
with $\ \varphi_{c,s}: [0,s_0+\cdots+s_n] \ \longrightarrow \ M\ $ the directed path associated to $\ (c,s) \in \Gamma_{p,q}(t) .$\\

\noindent \textbf{3. Functions on directed paths from Riemannian metrics on $M$.} \\

Let $g$ be a Riemannian metric on $M$.  We formally write
$$\int_{\Gamma_{p,q}(t)}e^{-l_g}\ dl \ \ = \ \  \sum_{n=1}^{\infty}\ \sum_{c\in D(n,k)}\ \int_{\Gamma_{p,q}^c(t)}e^{-l_g} \ dl_c,$$
where  $\ e^{-l_g}:\Gamma_{p,q}^c(t) \ \longrightarrow \ \mathbb{R} \ $ is the map given by $\e^{-l_g}(c,s) \  =  \  e^{-l_g(\varphi_{c,s})}\ $ and $\ l_g(\varphi_{c,s})\ $ is the length of the path $\ \varphi_{c,s},$ i.e.:
$$l_g(\varphi_{c,s})\  =  \ \sum_{i=0}^{l(c)}l_g(\varphi_{c,s}|_{[0,s_i]} ) \ =
\ \sum_{i=0}^{l(c)}\ \int_{0}^{s_i}g\Big( v_{c_i}(\varphi_{c_i}(p_{i}, u))\ , \ v_{c_i}(\varphi_{c_i}(p_{i}, u)) \Big)du.$$

\noindent \textbf{4. Functions on direct paths from functions on $M$.} \\

Given a smooth map $\ f:M \longrightarrow \mathbb{R}\ $  we formally write
$$\int_{\Gamma_{p,q}(t)}\widehat{f}\ dl \ \ = \ \  \sum_{n=1}^{\infty}\ \sum_{c\in D(n,k)}\ \int_{\Gamma_{p,q}^c(t)}f(p_0)\cdots f(p_{n}) \ dl_c$$
with $\ p_0=p \ \ \ \mbox{and} \ \ \ p_{i+1} = \varphi_{c_i}(p_i,s_i)\ \ $ for $\ \ 0\leq i \leq n$.\\

\noindent \textbf{5. Functions on directed paths from Lagrangian functions on $TM$.} \\

Let $\ L:TM \longrightarrow \mathbb{R} \ $ be a Lagrangian map.  In the applications $\ L \ $ is usually built from a Riemannian metric $\ g \ $ on $\ M \ $ and a potential map $\ U:M \longrightarrow \mathbb{R}\ $ as follows: $$L(p,v)\ = \ g(v,v) \ - \  U(p).$$
Given a Lagrangian $\ L \ $ we consider the following analogue of the Feynman integrals:
$$\int_{\Gamma_{p,q}(t)}e^{\frac{i}{\hbar}S}\ dl \ =  \  \sum_{n=1}^{\infty}\ \sum_{c\in D(n,k)}\ \int_{\Gamma_{p,q}^c(t)}e^{\frac{i}{\hbar}S} \ dl_c,$$
where we set $\ e^{\frac{i}{\hbar}S}(c,s) \ = \ e^{\frac{i}{\hbar}S(c,s)},\ \ $ and the action map $\ S\ $ is given by
$$S(c,s) \ =  \ \int_{0}^tL(\varphi_{c,s}(u), \dot{\varphi}_{c,s}(u))\ du \ =  \
\sum_{i=0}^{l(c)}\ \int_{0}^{s_i}L(\varphi_{c,s}|_{[0,s_i]}(u), \dot{\varphi}_{c,s}|_{[0,s_i]}(u))\ du .$$ This example both reveals the relations and differences between our constructions and Feynman integrals. Whereas in the latter arbitrary paths are taken into account, with our methods only paths with speeds and directions prescribed by the vector fields $\ v_1,...,v_k \ $ are allowed. Also, instead of looking for a measure on the space of all paths, we first decompose our space of paths into several pieces, and then impose a measure on each piece. Fortunately, each piece is finite dimensional and thus we have at our disposal the usual techniques coming from Riemannian geometry. Convergency of the sum of the integrals over each piece is to be studied in a case by case fashion.\\

\begin{rem}{\em In our examples we have found that the infinite sums defining the integrals above are actually convergent. Nevertheless, convergency is not a built-in property and should not be expected in general. To improve convergency properties one may look at the exponential generating series instead. For example, the $\mathrm{vol}$ function defined above can be  replaced by the function $\ \mathrm{vol}_{\lambda}, \ $ with $\ \lambda \ $ a positive real parameter, defined as follows:
$$\mathrm{vol}_{\lambda}(\Gamma_{p,q}(t)) \  =  \ \sum_{n=1}^{\infty}\Big( \sum_{c\in D(n,k)}\mathrm{vol}(\Gamma_{p,q}^c(t) \Big)\frac{\lambda^n}{n!} .$$
Clearly, this technique can be applied as well to the other quantities defined above. Moreover, if necessary, we may regard $\ \lambda \ $ as a formal parameter.}
\end{rem}

We have shown how to construct and integrate  functions on the moduli spaces of directed paths on directed manifolds. So let us pick one such a function and call it $\ g. \ $ Integrating over the moduli spaces of directed paths we obtain the kernel for the propagation of  influences  $\ k:M\times M \times \mathbb{R} \ \longrightarrow \ \mathbb{R}\ $ which is given by
$$k(p,q,t) \  = \  \int_{\Gamma_{p,q}(t)}g\ dl.$$

\begin{defn}{\em Let $\ (M,v_1,...,v_k) \ $ be a directed manifold with smooth moduli space of directed paths.  $M$ is given a Riemanninan metric, and thus it acquires a volume form. Let $\ f: M \longrightarrow \mathbb{R}\ $ be a map representing the density of influences originated at time $\ t=0. \ $ Let $\ g \ $ be a map on directed paths,  and consider its associated kernel of influences $\ k=k_g. \ $
The wave of influences $\ u: M\times \mathbb{R}_{\geq 0}\ \longrightarrow \ \mathbb{R} \ $ is the map  given by
$$u(q,t) \ =  \ \int_{p \in \Gamma_q^{-}(t)}k(p,q,t)f(p)\ dp,$$ where we assume that $\ \Gamma_q^{-}(t)\ $ is a compact oriented smooth sub-manifold
of $\ M; \ $ thus it acquires by restriction a Riemannian metric, and comes with a volume form $\ dp.$
}
\end{defn}

Let us consider a couple of examples.

\begin{itemize}

\item Let $\ g \ $ be the map constantly equal to $\ 1, \ $  we have that
$$u(q,t) \  =  \  \int_{p\in \Gamma_q^{-}(t)}\mathrm{vol}(\Gamma_{p,q}(t))f(p)\ dp.$$

\item For $\ g=e^{\frac{i}{\hbar}S}\ $ where $\ S \ $ is the action defined by a Lagrangian map, we have that
$$u(q,t) \  = \   \int_{p \in \Gamma_q^{-}(t)}\ \int_{\Gamma_{p,q}(t)}e^{\frac{i}{\hbar}S}f(p) \ dl  dp.$$

\end{itemize}

\section{Invariance, Involution, and Limit Properties}\label{iilp}

Let $\ (M, v_1,...,v_k)\ $ be a directed manifold and $\ f:M \longrightarrow N \ $ be a diffeomorphism. Then we obtain the directed manifold $$(N, f_{\ast}v_1,...,f_{\ast}v_k)$$ where the push-forward vector fields $\ f_{\ast}v_i\ $ are given for $\ q \in N \ $ by
$$f_{\ast}v_i(q) \ = \ d_pf(v_i(p)), \ \ \ \ \ \ \mbox{with} \ \ \ \ \ \  p\ = \ f^{-1}(q).$$ With this notation we have the following result.

\begin{thm}\label{ibd}{\em  Let $\ (M, v_1,...,v_k)\ $ be a directed manifold and $\ f:M \longrightarrow N\ $ be a diffeomorphism. For $\ p,q \in M \ $
the identity map gives a natural homeomorphism
 $$\Gamma_{p,q}^M(t) \  \simeq  \ \Gamma_{f(p),f(q)}^N(t).$$
Moreover, if $\ (M, v_1,...,v_k) \ $ has a smooth moduli space of directed paths, and $f$ is an orientation preserving diffeomorphism, then the identification above is an identity between Riemannian manifolds, and in particular
we obtain that $$\mathrm{vol}(\Gamma_{p,q}^M(t)) \  =  \ \mathrm{vol}(\Gamma_{f(p),f(q)}^N(t)).$$ }
\end{thm}

\begin{proof}We show that  $\ s \in \Gamma_{p,q}^{M,c}(t)\ \ $ if and only if $\ \ s \in \Gamma_{f(p),f(q)}^{N,c}(t).\ $ By construction we have that $$f(\varphi_{v_i}(p,t)) \  =  \ \varphi_{f_\ast(v_i)}(f(p),t),$$ and thus by induction on the length of $c$ we have that
$$f(\varphi_{v,c,s}(p,t)) \  =  \  \varphi_{f_\ast v, c,s}(f(p),t) ,$$ and therefore the equations
$$\varphi_{v,c}(p,s)\ =\ q \ \ \ \ \ \ \ \mbox{and} \ \ \ \ \ \ \ \varphi_{f_\ast v,c}(f(p),s)\ = \ f(q) \ \ \ \ \ \mbox{are equivalent.}$$
 For the second part we show that the identity map $\ \Gamma_{p,q}^{M,c}(t) \ \longrightarrow \ \Gamma_{f(p),f(q)}^{N,c}(t)\ $ preserves orientation.
 Since the identity map preserves the splittings
$$T_s\Delta_n^t \ \ \simeq \ \ N_s\Gamma_{p,q}^{M,c}(t) \ \oplus \ T_s\Gamma_{p,q}^{M,c}(t)   \ \ \ \ \ \ \mbox{and} \ \ \ \ \ \ \
T_s\Delta_n^t \ \ \simeq \ \  N_s\Gamma_{f(p),f(q)}^{N,c}(t) \ \oplus \ T_s\Gamma_{f(p),f(q)}^{N,c}(t)  ,$$ we just have to show
that $\ N_s\Gamma_{p,q}^{M,c}(t)\ \ $ and $\ \ N_s\Gamma_{f(p),f(q)}^{N,c}(t)\ $ are given compatible orientations. This follows by
construction, see the proof of Theorem \ref{sii}, as the square
\[\xymatrix @R=.4in  @C=.8in
{N_s\Gamma_{p,q}^c(t) \ar[r]^{d_s\phi_{v,c}} \ar[d]_{1} & T_{\phi_{v,c}(s)}M \ar[d]^{df}
\\ N_s\Gamma_{f(p),f(q)}^{N,c}(t) \ar[r]^{d_s\phi_{f_{\star}v,c}} & T_{\phi_{f_{\star}v,c}(s)}N   } \]
is a commutative diagram of orientation preserving isomorphisms, see Corollary \ref{oc}.

\end{proof}

Next result tell us how the moduli spaces of directed paths depend on the ordering on vector fields.

\begin{prop}\label{lis}
{\em  Let $\ (M, v_1,...,v_k)\ $ be a directed manifold and $\ \alpha: [k] \longrightarrow [k]\ $ be a permutation. For the directed manifold
$\ (M, v_{\alpha 1},...,v_{\alpha k})\ $ we have that
 $$\Gamma_{p,q}^{v}(t) \ \simeq  \ \Gamma_{p,q}^{v\alpha}(t).$$
Moreover, if $\ (M, v)\ $ has a smooth moduli space of directed paths, then so does $\ (M, v\alpha)\ $ and
we have that $$\mathrm{vol}(\Gamma_{p,q}^{v}(t)) \ =  \ \mathrm{vol}(\Gamma_{p,q}^{v\alpha}(t)) .$$ }
\end{prop}

\begin{proof}We regard the permutation $\ \alpha\ $ as a map $$\alpha_{\ast}: \Gamma_{p,q}^{v}(t) \  \longrightarrow  \ \Gamma_{p,q}^{ v \alpha}(t) \ \ \ \ \ \ \mbox{given by}\ \ \ \ \ \ \alpha_{\ast}(c,s) \ = \  (\alpha^{-1}c,s).$$ It follows that  $\ \alpha\ $ is an homeomorphism as its restriction map
$$\alpha_{\ast}:  \Gamma_{p,q}^{v,c}(t)  \  \longrightarrow  \ \Gamma_{p,q}^{v \alpha, \alpha^{-1} c}(t) $$
is just the identity map and is a well-defined homeomorphism since
$$\varphi_{ v \alpha,\alpha^{-1} c}(p,s) \  =  \ \varphi_{v, \alpha \alpha^{-1} c}(p,s) \ =  \ \varphi_{v,c}(p,s) \  =  \ q .$$ In the case of a smooth moduli space of directed paths, the map above is clearly orientation preserving, since it is just the identity map, and we have a commutative
diagram of orientation-preserving isomorphisms
\[\xymatrix @R=.4in  @C=.8in
{N_s\Gamma_{p,q}^{v,c}(t) \ar[r]^{d_s\phi_{v,c}} \ar[d]_{1} & T_{\phi_{v,c}(s)}M \ar[d]^{1}
\\ N_s\Gamma_{p,q}^{v\alpha, \alpha^{-1}c}(t) \ar[r]^{d_s\phi_{v\alpha, \alpha^{-1}c}} & T_{\phi_{v\alpha, \alpha^{-1}c}(s)}M  } \]
\end{proof}

From the Theorem \ref{ibd} and Proposition \ref{lis} we see that the invariant study of directed paths on a directed oriented manifold $\ M \ $ relies on the study,  for $\ k\geq 1, \ $  of the quotient spaces
$$\chi(M)^k / \mathrm{Diff}_+(M) \times S_k,$$ where  $\ \chi(M) \ $ is the space of vector fields on $\ M, \  $ $\ S_k\ $ the group of permutations of $\ [k], \ $ and $\ \mathrm{Diff}(M)_+ \ $ is the group of orientation preserving diffeomorphism of $\ M, \ $ i.e. the study of equivalence classes of tuples of vector fields under diffeomorphisms and permutations.\\

Next we define the direction reversion functor $\ -:\mathrm{diman} \longrightarrow \mathrm{diman}. \ $ It sends
a directed manifold $\ (M, v_1,...,v_k)\ $ to its reversed directed manifold $$ (M, -v_1,...,-v_k).$$

\begin{prop}{\em Let $\ (M, v_1,...,v_k)\ $ be a directed manifold and $\ (M, -v_1,...,-v_k) \ $ its reversed directed manifold. We have canonical homeomorphisms
$$\Gamma_{v,A,B}(t) \  \simeq \ \Gamma_{-v,B,A}(t).$$
And therefore the respective reachable sets are related by:
$$\Gamma_{-v,A}(t) \ \simeq \ \Gamma^{-}_{v,A}(t), \ \ \ \  \ \Gamma_{-v, A,\leq}(t) \ \simeq \ \Gamma^{-}_{v, A,\leq}(t),  \ \ \ \  \ \Gamma_{-v,A}  \  \simeq \ \Gamma^{-}_{v,A},
\ \ \ \  \ \mathrm{F}_{-v,A}(t)  \ \simeq \  \partial \Gamma^{-}_{A,\leq}(t).$$
If $\ (M, v_1,...,v_k)\ $ has a smooth moduli space of directed paths, then so does $\ (M, -v_1,...,-v_k)\ $ and the maps above are actually diffeomorphisms. These diffeomorphisms may or may not preserve orientation.

}
\end{prop}
\begin{proof}We define a map $\ \overline{(\ )}:\Gamma_{v,A,B}(t) \ \longrightarrow \  \Gamma_{-v,B,A}(t)\ $ as follows
$$\overline{(c,s)}\  =  \ \overline{(c_0,...,c_n,s_0,...,s_n)} \  =  \ (\overline{c},\overline{s}) \ =  \ (c_n,...,c_0,s_n,...,s_0).$$ This map is an homeomorphism since the map
$$\overline{(\ )}:D(n,k) \ \longrightarrow \  D(n,k)$$ is bijective, and the map
$$\overline{(\ )}:\Gamma^c_{v,A,B}(t) \ \longrightarrow \  \Gamma^{\overline{c}}_{-v,B,A}(t)$$ is an homeomorphism
as the equations
$$\varphi_{v,c}(p,s)\ = \ q   \ \ \ \ \ \ \ \mbox{and} \ \ \ \ \ \ \ \varphi_{-v,\overline{c}}(q,\overline{s})\ = \ p \ \ \ \ \ \mbox{are equivalent}.$$
\end{proof}

In quantum mechanics the proposed integration domain of a Feynman integral is usually the space of differentiable paths,  with fixed endpoints,  on a manifold. We think of the moduli spaces of directed paths $\ \Gamma_{p,q}(t) \ $ as being analogues for the integration domains for Feynman path integrals, where in addition to boundary restrictions, we impose tangential restrictions on the allowed paths; these restrictions induce a partition of path-space into finite dimensional pieces. The question arises: Can we somehow approach the full Feynman domains of integration from the moduli spaces of directed paths? In other words, is it  possible to relax our definition of directed paths, or  perform some kind of limit procedure that allow us to approach Feynman integrals from the viewpoint of indirect influences? We left this problem open for future research, and limit ourselves to make a couple of remarks along this line of thinking.\\

Clearly what one should do is to allow more paths into our moduli spaces. One way to go is to replace the vector fields $\ v_j \ $ by sections of the projective tangent bundle $\ \mathbb{P}TM, \ $ so that one fixes the directions along which our curves can move, but leave the speeds unconstrained. Although this approach may be of interest, finite dimensionality is lost. Incidentally, this approach establishes the connection with directed topological spaces \cite{grandis}. \\

Instead we propose another approach. Given a directed manifold $$ (M, v)\ = \ (M,v_1,....,v_k)$$  we consider the tuple $\ v(a,b) \ $  of vector fields on $\ M, \ $ for $\ a,b \in \mathbb{N}_+, \ $
given by the lexicographically ordered set:
$$v(a,b) \  =  \ \{ \ \frac{i}{b}v_j\ | \ -ab \leq i \leq ab, \  \ j \in [k]\ \}. $$ Indirect influences on the directed manifold $\ (M, v(a,b))\ $ are exerted trough paths along the directions defined by the vector fields $\ v_j \ $ with rather arbitrary speeds,  if $\ a \ $ and $\ b \ $ are large numbers. Piecewise finite dimensionality is preserved for $\ a\ $ and $\ b\ $ fix.\\

To relax even further the restrictions on the paths in our moduli spaces we consider directed manifolds of the form $\ (M, <v(a,b)>)\ $ where in $\ <v(a,b)> \ $ we include all vector fields that are finite sums of vector fields in $\ v(a,b). \ $ Indirect influences in $\ (M, <v(a,b)>)\ $ are exerted trough paths with rather arbitrary speeds and directions; for example, if the vector fields in $\ v(a,b) \ $ at some point contain a basis of the tangent space, then essentially all directions and speeds are allowed, for $\ a \ $ and $\ b \ $ large, at that point. Piecewise finite dimensionality is preserved for $\ a \ $ and $\ b \ $ fix.\\

The fundamental question is whether it is possible to make any sense of the limit of the moduli spaces of directed path for the spaces $\ (M, <v(a,b)>)\ $ as $\ a \ $ and $\ b \ $ grow to infinity, a question however beyond the scope of this work.

\section{Indirect Influences on Product/Quotient Manifolds}\label{iipqm}

Let $\ (M, v_1,...,v_k)\ \ $ and  $\ \ (N, u_1,...,u_l)\ $ be directed manifolds. The natural isomorphism $$T(M\times N) \ \ \simeq \ \ \pi_M^{\ast}TM \ \oplus \ \pi_N^{\ast}TN,$$ allows us to consider $$(M\times N, v_1,...,v_k, u_1,...,u_l)\ \ \ \mbox{as a directed manifold,} $$ where one should more formally write $\ (v_i,0)\ $ instead of $\ v_i, \ $ and $\ (0,u_j)\ $ instead of $\ u_j.$ \\

Let $\ \mathrm{diman} \ $ be the category of directed manifolds. We allow in  $\mathrm{diman}$  manifolds with connected components of different dimensions, and assume by convention
that the set with one element is a  directed manifold.

\begin{prop}{\em The product defined above gives $\mathrm{diman}$ the structure of a monoidal category with unit the set $\ [1].$
}
\end{prop}

Fix $\ A \ \subseteq \ [n]. \ $ We say that a map $\ c:A \longrightarrow [k]\ $ is a pattern if $\ c(i)\  \neq \  c(i+1)\ $ for all contiguous elements $\ i,\ i+1 \in A.\ $
Thus a pattern for the product manifold $\ M\times N\ $ is given by a map
$\ \ c:[n]\ \longrightarrow \ [k] \sqcup [l]  \ \ \ \mbox{such that its restrictions}$
$$\ \ \ \ c|_{c^{-1}[k]}: c^{-1}[k]\ \longrightarrow \ [k]\ \ \ \ \ \ \ \ \mbox{and} \ \ \ \ \ \ \ \ c|_{c^{-1}[l]}: c^{-1}[l] \ \longrightarrow \ [l]\ \ \ \ $$ are patterns on $\ c^{-1}[k]\ \ $ and $\ \ c^{-1}[l], \ $ respectively.

\begin{prop}{\em Let $\ (p_1,p_2),\ (q_1,q_2) \ \in \ M\times N, \ $ and let $\ c:[n] \longrightarrow [k] \sqcup [l]\ $ be a pattern. We have a canonical homeomorphism: $$\Gamma_{(p_1,p_2),(q_1,q_2)}^{M\times N,\ c} \ \ \simeq \ \ \Gamma_{p_1,\ q_1}^{N,\ c|_{c^{-1}[k]} } \ \times \ \Gamma_{p_2,\ q_2}^{N,\ c|_{c^{-1}[l]} } .$$ }
\end{prop}

\begin{proof} The desired homeomorphism sends
$$ s \ \in \ \Gamma_{(p_1,p_2),(q_1,q_2)}^{M\times N,\ c}(t)\ \ \ \subseteq \ \ \ \Gamma_{(p_1,p_2),(q_1,q_2)}^{M\times N,\ c}\ \ $$ to the pair
$$(s|_{c^{-1}[k]}, \ s|_{c^{-1}[l]}) \ \in \ \Gamma_{p_1,\ q_1}^{N,\ c|_{c^{-1}[k]} } (a) \ \times \ \Gamma_{p_2,\ q_2}^{N,\ c|_{c^{-1}[l]} } (t-a),$$
where $$a \ =  \ \sum_{i \in c^{-1}[k]}s_i .$$
\end{proof}

Next we consider the moduli spaces of directed paths on quotient manifolds. Let $\ M \ $ be a smooth manifold, $\ G \ $ a compact Lie group acting freely on $\ M, \ $ and assume that the directed manifold $\ (M, v_1,...,v_k)\ $ is invariant under the action of $\ G, \ $ i.e.
the following identities hold:
$$d_pg(v_i) \ =  \ v_i(gp) \ \ \ \ \ \mbox{for all} \ \ \ p\in M, \ \ g \in G.$$
Then $\ M/G \ $ is a smooth manifold and it comes with a smooth quotient map
$$\pi: M \ \longrightarrow \ M/G ,$$ which induces a surjective map $\ d\pi: TM \ \longrightarrow \ T(M/G) , \ $ and canonical isomorphisms
$$\overline{d_p\pi}: T_pM/T_p(Gp) \  \longrightarrow  \ T_{\overline{p}}(M/G).$$
Note also that we have isomorphisms
$$T_{\overline{p}}(M/G) \ \ \simeq \ \ \Big(\bigoplus_{g\in G}T_{gp}M \Big)/G.$$
Thus we obtain the directed manifold $\ (M/G, \  \overline{v}_1,\ ... \ , \overline{v}_k)\ \ $ with $\ \ \overline{v}_i \ = \ d\pi(v_i)$.

\begin{thm}{\em Let $\ (M, v_1,...,v_k)\ $ be a directed manifold, invariant under the action of the compact Lie group $\ G, \ $
and let $\ p,q \in M. \ $ Then $\ (M/G, \  \overline{v}_1,\ ... \ , \overline{v}_k)\ \ $ with $\ \ \overline{v}_i \ = \ d\pi(v_i)$ is a directed
manifold, $\ G \ $ acts naturally on $ \ \Gamma_{G_p,G_q}^{M}(t), \ $ and we have that
$$\Gamma_{\overline{p}, \overline{q}}^{M/G}(t) \ \ \simeq \ \ \Big(\Gamma_{G_p,G_q}^{M}(t)\Big)/G .$$ }
\end{thm}

\begin{proof}The  result follows from the fact that there are $G$-equivariant homeomorphisms
$$\Gamma_{p, Gq}^{M}(t) \ \ \longrightarrow \ \ \Gamma_{\overline{p}, \overline{q}}^{M/G}(t) \ \ \ \ \ \ \  \mbox{and} \ \ \ \ \ \ \
\Gamma_{p, Gq}^{M}(t) \ \ \longrightarrow \ \ \Big(\Gamma_{Gp,Gq}^{M}(t)\Big)/G. $$
As the vector fields $\ v_i\ $ are $\ G$-invariant, the corresponding flows $\ \varphi_i \ $ are also $\ G$-invariant:
$$\varphi_i(gp,t) \  =  \ g\varphi_i(p,t), \ \ \ \ \ \mbox{and therefore} \ \ \ \ \ \varphi_{c,s}(gp,t) \  =  \ g \varphi_{c,s}(p,t) $$
for any pattern and time distribution $\ (c,s). \ $ This shows that $\ G\ $ acts on $\ \Gamma_{Gp,Gq}^{M}(t), \ $ and  that $\ \Gamma_{p,q}^M(t)\  \simeq  \ \Gamma_{gp,gq}^M(t) \ \ \ \ \mbox{for}  \ \ \ p,q \in M. \ $ \\

A pair $(c,s)$ defines a directed path from $\ \overline{p}\ $ to $\ \overline{q}\ $ in $\ M/G\ $ if and only if
$\ \ \overline{\varphi}_{c}(\overline{p},s) \ = \ \overline{q}.\ $ If the latter equation holds we have that
$$\pi \varphi_{c}(p,s) \  =  \ \overline{\varphi}_{c}(\overline{p},s) \  =  \ \overline{q} \ \ \ \ \ \mbox{and \ \ thus} \ \ \ \ \ \varphi_{c}(p,s) \in Gq. $$
Therefore $\ (c,s)\ $ defines a directed path from  $\ \overline{p}\ $ to $\ \overline{q}\ $  if and only if $\ (c,s) \ $ defines an indirect influence from $\ p \ $ to $\ Gq.\ $ So we have shown that the map $\ \Gamma_{p, Gq}^{M}(t) \ \longrightarrow \   \Gamma_{\overline{p}, \overline{q}}^{M/G}(t)\ $ is a $G$-equivariant homeomorphism.\\

Similarly,  if $\ a \in Gp, \ $ then $\ \ \overline{\varphi}_{c}(\overline{p},s) \ = \ \overline{q}\ \ $ if and only if
$$\varphi_{c}(a,s)\  =  \ \varphi_{c}(gp,s)\  =  \ g\varphi_{c}(p,s) \ \ \ \ \ \mbox{belongs to } \ \ Gq.$$ Thus the map
$\ \Gamma_{p, Gq}^{M}(t) \ \longrightarrow \ \Big(\Gamma_{G_p,G_q}^{M}(t)\Big)/G \ $ is a $\ G$-equivariant homeomorphism.\\

\end{proof}

\section{Directed Paths for Constant Vector Fields}\label{iicvf}

As a first and pretty workable example, linking the theory of indirect influences on directed manifolds with linear programming techniques, we consider constant vector fields on affine spaces. Thus we fix a directed manifold $\ (\mathbb{R}^d, v_1,...,v_k) \ $ where the vector fields
$$v_j\  =  \ \sum_{j=1}^d a_{ij} \parcial{}{x_i},$$ have constant coefficients $\ a_{ij} \in \mathbb{R}\ \ $   for $\ \ i \in [d], \ $ $\ j\in [k]. \ $

\begin{thm}\label{t1}{\em Consider the directed manifold  $\ (\mathbb{R}^d, v_1,...,v_k). \ $
Fix a pattern $\ c \in D(n,k)\ $ and points $\ p,q \in \mathbb{R}^d. \ $ The space of directed paths $\ \Gamma_{p,q}^{c}(t) \ $ is the convex polytope given on the variables $\ s \in \mathbb{R}_{\geq 0}^{n+1}\ $ by the system of equations:
$$a_{ic(0)}s_{0} \ + \ \cdots \ + \ a_{ic(n)}s_n \  =   \ q_i-p_i, \ \ \ \mbox{for} \ \ i \in [d], \ \ \ \   \mbox{and}  \ \ \ \ s_0 \ + \ \cdots \ + s_n \  =  \ 1,$$
or equivalently in matrix notation
$$\left(
    \begin{array}{c}
      A_c \\
      1 \\
    \end{array}
  \right) s  \  =  \
  \left(
    \begin{array}{c}
      q-p \\
      t \\
    \end{array}
    \right),
 $$
where $\ A_c \ $ is the matrix of format $\ d\times (n+1)\ $ given by:
$$\ (A_c)_{ij}\ = \ a_{ic(j)}, \ \ \ 1=(1,...,1) \in \mathbb{R}^{n+1},$$
$$s \ = \ (s_0,...,s_n), \  \ p \ = \ (p_1,...,p_d),\ \ \
\mbox{and} \ \ \ q \ = \ (q_1,...,q_d).$$}
\end{thm}

\begin{proof}The result follows from the fact that the solutions of the differential equation $\ \dot{p}\ = \ v, \ $ where $\ v \ $ is constant and with initial condition $\ a, \ $ are of the form $\ p(t)\ = \ a+tv.$
\end{proof}

\begin{thm}\label{t2}{\em Consider the directed manifold  $\ (\mathbb{R}^d, v_1,...,v_k). \ $ For $\ p,q \in \mathbb{R}^d , \ $ the volume of
the space of directed paths $\ \Gamma_{p,q}^{c}(t)\ $ is given by
$$ \mathrm{vol}(\Gamma_{p,q}^{c}(t)) \ = \ \mathrm{vol}(\mathrm{Conv}(u_I)), \ \ \ \ \mbox{where:}$$
$\mathrm{Conv}(u_I)$ is the convex hull of the vector $\ u_I\ $ defined by the following conditions:
\begin{itemize}
  \item $I \subseteq [d]\ $ is a subset of cardinality $\ n$.
  \item The entries of the  vector $\ u_I \in \mathbb{R}_{\geq 0}^{n+1}\ $ vanish for indexes not in $\ I$.
  \item For a matrix $\ A \ $ we let $\ A_I\ $ be its restriction to the columns with indexes in $\ I. \ $ The set $\ I \ $ must be such that
  $$\mathrm{det}\left(
    \begin{array}{c}
      A_c \\
      1 \\
    \end{array}
  \right)_I   \ \ \ \neq \ \ \ 0.$$
  \item $u_I\ $ is the unique solution of the linear system:
  $$\left(
    \begin{array}{c}
      A_c \\
      1 \\
    \end{array}
  \right)_I u_I  \ \ \  = \ \ \
  \left(
    \begin{array}{c}
      q-p \\
      t \\
    \end{array}
    \right).$$
\end{itemize}
 }
\end{thm}

\begin{proof}Theorem \ref{t1} and standard results of linear programming \cite{Optimization1, Optimization2} one can show that
$\ \Gamma_{p,q}^{c}(t) \ = \ \mathrm{Conv}(u_I).$
\end{proof}

\subsection{Dimension One}

Consider the directed manifold $\ (\mathbb{R},\  a_1\frac{d}{d x},\ ...\ ,\ a_k\frac{d}{d x})\ $  where for simplicity we assume that $\ a_i \ \neq \ a_j. \ $  Fix a pattern $\ c \in D(n,k)\ $ and consider the space $\ \Gamma_{0,x}^{c}(t)\ $ of directed paths  from $\ 0 \ $ to $\ x \ $ exerted in time $\ t. \ $ The space $\ \Gamma_{0,x}^{c}(t)\ \subseteq  \ \mathbb{R}_{\geq 0}^{n+1} \ \ $  is the convex polytope defined  by the equations
$$a_{c(0)}s_0 \ + \ \cdots \ +\  a_{c(n)}s_n \ = \ x \ \ \ \ \ \ \ \mbox{and} \ \ \ \ \ \ \  s_0 \ + \ \cdots \ + \ s_n \ = \ t.$$
Consider the set $$D\ = \ \{\ (i , j)  \in  [n] \ \ | \ \ c(i) \ \neq \ c(j)\ \}.$$  By Theorem
$\ \Gamma_{0,x}^{c}(t)\ $  is  the convex polytope $ \ \mathrm{Conv}(u_{ij})\ $ generated by the vectors
$\ u_{ij},\ $  given for $\ (i,j) \in D \ $ by $$u_{ij} \ \ = \ \ (0, \ldots, \underset{i\uparrow}{l_i}, \ldots, 0 , \ldots, \underset{j\uparrow}{l_j}, \ldots, 0)\ \in \ \mathbb{R}_{\geq 0}^{n+1}$$ where
$$\left(
    \begin{array}{cc}
      a_{c_i} & a_{c_j} \\
      1 & 1 \\
    \end{array}
  \right)
 \left(
   \begin{array}{c}
     l_i \\
     l_j \\
   \end{array}
 \right) \ \ = \ \ \left(
                     \begin{array}{c}
                       x \\
                       t \\
                     \end{array}
                   \right)
 .$$

 \

Below we use the following identity, valid for $\ n,m \in \mathbb{N}, \ $ involving the classical beta $\ \mathrm{B}  \ $ and gamma $\ \Gamma \ $ functions:
$$\int_0^1s^n(1-s)^mds\ = \ \mathrm{B}(n+1,m+1)\ = \ \frac{\Gamma(n+1)\Gamma(m+1)}{\Gamma(n+m+2)}
\ = \ \frac{n! m!}{(n+m+1)!}.$$

\begin{thm}{\em Consider the directed manifold $\  (\mathbb{R}, \ \frac{d}{d x} , \ -\frac{d}{d x}). \ $ For $\ x,y \in \mathbb{R} \ $ we have that
$\ \mathrm{vol}(\Gamma_{0,x}(t)) =0\ $ if $\ |x|>t, \ $  $\ \mathrm{vol}(\Gamma_{0,x}(t)) =1\ $ if $\ |x|=t, \ $ and otherwise is given by:
$$\sum_{n=0}^\infty \Big[ \frac{(t+x)^{n} (t-x)^{n}}{n!^2 }\ + \
2t\frac{(t+x)^{n}(t-x)^{n}}{(n+1)! n!}\Big]2^{1-2n}.$$
Furthermore, we have that $$\mathrm{vol}(\Gamma_{x,0}(t)) \ = \ \mathrm{vol}(\Gamma_{0,x}(t)) \ \ \ \ \ \ \mbox{and} \ \ \ \ \ \
\mathrm{vol}(\Gamma_{x,y}(t)) \  =  \ \mathrm{vol}(\Gamma_{0,y-x}(t)). $$
The wave of influences for $\ t>0 \ $ is given by $$u(x,t)\  =  \ \int_{x-t}^{x+t} \mathrm{vol}(\Gamma_{y,x}(t))dy \ \ \ \ \  \mbox{is constant in}  \ \ \ x \in \mathbb{R},$$ and is given explicitly by $ \ \ u(x,t)\  = \   10e^t + 6e^{-t} - 16.$
}
\end{thm}

\begin{proof} Fix $\ x\in \mathbb{R} \ $ and a pattern $\ c \in D(n,k). \ $ The space of directed paths
$\ \Gamma_{0,x}^{c}(t)\ $ is the polytope given by
$$\sum_{i=0}^n (-1)^{c_i} s_{i}\ =  \ x \ \ \ \ \ \ \ \ \mbox{and} \ \ \ \ \ \ \ \ \ \ \sum_{i=0}^n s_i\  =  \ t.$$
Since we have just two vector fields, a pattern $\ (c_0,...,c_n) \ $ is determined by its initial value $\ c_0. \ $  Figure 4 shows
the directed path associated to  the tuple   $ \  (7,5,3,7) \in \Gamma^{(1,2,1,2)}_{(0,-2)}.$
\begin{center}
\psset{unit=0.5cm}
\begin{pspicture}(-20,1)(2,7)
        \psline[linewidth=1pt]{}(-20,4)(2,4)
        \rput(-9,3.6){$0$}
        \psdot[dotstyle=*,
    dotsize=2pt](-9,4)
    \psline[linewidth=2pt]{->}(-9,4)(-2,4)
    \rput(-2,3.6){$7$}
    \psline[linewidth=2pt]{->}(-2,4.3)(-7,4.3)
    \rput(-7,3.6){$2$}
    \psline[linewidth=2pt]{->}(-7,4.6)(-4,4.6)
    \rput(-4,3.6){$5$}
    \psline[linewidth=2pt]{->}(-4,4.9)(-12,4.9)
    \rput(-12,3.6){$-2$}
    \rput(-9,1.5){Figure 4: Directed path associated to the tuple $\ (7,5,3,7)\  \in \ \Gamma^{(1,2,1,2)}_{(0,-2)}$.}
\end{pspicture}
\end{center}
We distinguish  four cases taking into account the initial value $\ c_0 \ $ and the parity of $\ n. \ $\\

\noindent Consider the pattern $\  (1,2, ... , 1,2)\ $  of length $\ 2n, \ $ for $\ n\geq 1.$
Then $\ \Gamma_{0,x}^{c}(t)\ $ is the polytope given by
$$\sum_{i=0}^{2n-1} (-1)^{i} s_i\  =  \ x \ \ \ \ \ \ \  \mbox{and} \ \ \ \ \ \ \ \sum_{i=0}^{2n-1} s_i\ =  \ t.$$
Setting $$\sum_{i=0}^{n-1}  s_{2i}\  =  \ a \ \ \ \ \ \ \ \ \mbox{and} \ \ \ \ \ \ \ \sum_{i=0}^{n-1}  s_{2i+1} \ =  \ b,$$
the  previous equations become  $\ a-b\ = \ x\ \ $ and $\ \ a+b\ = \ t, \ $ with solutions
$\ a\ = \ \frac{t+x}{2}\ \ $ and $\ \ b\ = \ \frac{t-x}{2}. \ $ By definition $\ a,b \geq 0, \ $
thus we must have $\ \normv{x} < t \ $ in order that $\ \Gamma_{0,x}^{c}(t) \ \neq \ \emptyset. \ \ $ For $\ \normv{x} < t, \ $
we have  that  $$\Gamma_{0,x}^{c}(t)\  =  \ \Delta_{n-1}(\frac{t+x}{2})\ \times \ \Delta_{n-1}(\frac{t-x}{2}),$$ and therefore
$$\mathrm{vol}(\Gamma_{0,x}^{c}(t))\  =  \ \frac{(t+x)^{n-1} (t-x)^{n-1}}{2^{2n-2}(n-1)!^2}.$$

\noindent For the pattern $\ (1,2, ... , 1,2,1)\ $ of length $\ 2n+1, \ $ with $\ n \geq 1, \ $
setting $$\sum_{i=0}^{n}  s_{2i}\  =  \ a  \ \ \ \ \ \  \mbox{and} \ \ \ \ \ \ \sum_{i=0}^{n-1}  s_{2i+1}\  =  \ b \ \ \ \ \ \mbox{we \ get \ that}$$
$$\mathrm{vol}(\Gamma_{0,x}^{c}(t))\  = \ \mathrm{vol}\Big[ \Delta_{n}(\frac{t+x}{2})\ \times \ \Delta_{n-1}(\frac{t-x}{2}) \Big] \  =  \
\frac{(t+x)^{n} (t-x)^{n-1}}{2^{2n-1}n!(n-1)!}.$$

\noindent The pattern $\ c=(2,1, \cdots, 2,1)\ $ of length $\ 2n, \ $ with $\ n
\geq 1,\ $ leads to
 $$\mathrm{vol}(\Gamma_{0,x}^{c}(t))\ \ = \ \ \mathrm{vol}\Big[ \Delta_{n-1}(\frac{t-x}{2})\ \times \ \Delta_{n-1}(\frac{t+x}{2}) \Big] \ \ = \ \
 \frac{(t+x)^{n-1} (t-x)^{n-1}}{2^{2n-2}(n-1)!^2}.$$

\noindent For the pattern  $\ c=(2,1, \cdots,2,1,2)\ $ of length $\ 2n+1, \ $ with $\ n \geq 1, \ $ we get that
$$\mathrm{vol}(\Gamma_{0,x}^{c}(t))\ \ = \ \ \mathrm{vol}\Big[ \Delta_{n}(\frac{t-x}{2})\ \times \ \Delta_{n-1}(\frac{t+x}{2}) \Big] \  =   \  \frac{(t+x)^{n-1} (t-x)^{n}}{2^{2n-1}(n-1)!n! }.$$
Therefore $\ \mathrm{vol}(\Gamma_{0,x}(t)) \ $  is for $\ |x|<t \ $ given by:
$$\sum_{n=1}^\infty \Big[ \frac{(t+x)^{n-1} (t-x)^{n-1}}{(n-1)!^2 }\ + \  \frac{(t+x)^n (t-x)^{n-1}}{n! (n-1)!} \ + \  \frac{(t+x)^{n-1}(t-x)^n }{n! (n-1)!} \Big]2^{1-2n}$$ yielding the desired result. \\

Applying Theorem \ref{ibd} to translations on $\ \mathbb{R}\ $ we obtain that:
$$\mathrm{vol}(\Gamma_{x,y}(t)) \  =  \ \mathrm{vol}(\Gamma_{x-x,y-x}(t)) \  =  \ \mathrm{vol}(\Gamma_{0,y-x}(t)). $$ In particular we get that $\ \mathrm{vol}(\Gamma_{x,0}(t)) \ = \ \mathrm{vol}(\Gamma_{0,-x}(t)). \ $  A direct inspection of the explicit formula for $\ \mathrm{vol}(\Gamma_{x,y}(t))\ $ given above yields $\ \mathrm{vol}(\Gamma_{0,-x}(t)) \ = \ \mathrm{vol}(\Gamma_{0,x}(t)).\ $\\

Next we show that the wave of influences is constant in the variable $x$. Making the change of variables $\ y -x \ \rightarrow \ y\ $ we get that:
$$u(x,t)\  =  \ \int_{x-t}^{x+t} \mathrm{vol}(\Gamma_{0,x-y}(t))dy \  =  \ \int_{-t}^{t} \mathrm{vol}(\Gamma_{0,-y}(t))dy
 \  =  \ \int_{-t}^{t} \mathrm{vol}(\Gamma_{0,y}(t))dy \  =  \ u(0,t).$$
To compute $\ u(0,t)\ $ we make the change of variable $\ y\ = \ t(2s-1) \ $ in the  integral
$$\int_{-t}^{t} \sum_{n=0}^\infty \Big[ \frac{(t+y)^{n} (t-y)^{n}}{n!^2 }\ + \
2t\frac{(t+y)^{n}(t-y)^{n}}{(n+1)! n!}\Big]2^{1-2n}dy \ \ = $$

$$ \sum_{n=0}^\infty\ 4\frac{t^{2n+1}}{n!n!} \int_{0}^{1}s^n(1-s)^{n} ds
\ + \  8\frac{t^{2n+2}}{(n+1)! n!}\int_{0}^{1}s^n(1-s)^{n}ds  \ \  =  $$

$$ 4\sum_{n=0}^\infty\ \frac{t^{2n+1}}{(2n+1)!}
\ + \ 16\sum_{n=0}^\infty \ \frac{t^{2n+2}}{(2n+2)! }  \ \  = $$
$$4\mathrm{sinh}(t) \ + \ 16(\mathrm{cosh}(t)-1) \ = \  10e^t + 6e^{-t} - 16.  $$

\end{proof}

\subsection{Dimension Two}

Consider the directed manifold $\  (\mathbb{R}^2, \ \frac{\partial}{\partial x}, \ \frac{\partial}{\partial y}),\ $ and let
$\ \Gamma(x,y)\ = \ \Gamma_{(0,0),(x,y)} \ $ be the moduli space of directed paths from $\ (0,0)\ $ to $\ (x,y).\ $  Note that such influences can only happen at time $\ t=x+y,\ $ and thus there is no need to include the time variable in the notation. Figure 4 shows the directed path associated to  the tuple
$\ (1,3,2,1) \in  \Gamma^{(2,1,2,1)}(4,3).$

\

\begin{center}
      \psset{unit=0.6cm}
    \begin{pspicture}(-1,-1)(5,4)
    \psgrid[subgriddiv=1,griddots=8,gridlabels=8pt](-1,-1)(5,4)
              \psline[linewidth=2pt]{->}(0,0)(0,1)
              \psline[linewidth=2pt]{->}(0,1)(3,1)
              \psline[linewidth=2pt]{->}(3,1)(3,3)
              \psline[linewidth=2pt]{->}(3,1)(3,3)
              \psline[linewidth=2pt]{->}(3,3)(4,3)
              \rput(0,-0.5){$(0,0)$}
              \rput(4.7,3){$(4,3)$}
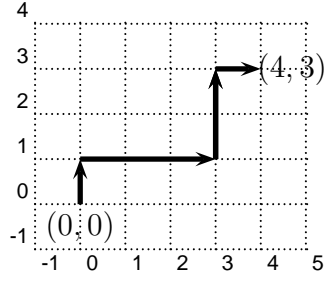
              \rput(3,-2.5){Figure 4. Directed path associated to  $\ (1,3,2,1)\ \in \ \Gamma^{(2,1,2,1)}(4,3).$}
        \end{pspicture}
\end{center}

\

\

In our next results we use the following notation. For $\ k \in \mathbb{N}\ $ we set
$$i_k(x,y) \ = \ \sum_{n=0}^{\infty}\frac{x^ny^{n+k}}{n!(n+k)!} \ \ \ \ \ \mbox{and} \ \ \ \ \
i_{-k}(x,y) \ = \ i_k(y,x) \ = \ \sum_{n=0}^{\infty}\frac{x^{n+k}y^n}{(n+k)!n!}.$$

The following result is easy to check.

\begin{lem}\label{sal}{\em \
\begin{itemize}
  \item For $\ l,m \in \mathbb{N} \ $ and $\ k \in \mathbb{Z} \ $ we have  that
$$\frac{\partial^l}{\partial x^l}\frac{\partial^m}{\partial y^m}i_k(x,y) \ = \ i_{l-m+k}(x,y).$$

  \item For $\ k \in \mathbb{N},\ $ the function $\ i_k(x,y)\ $ is given in terms of the modified Bessel function $\ I_k(z) \ $ by
  $$i_k(x,y) \ = \ x^{-\frac{k}{2}}y^{\frac{k}{2}}I_k(2\sqrt{xy}) ,$$ where we recall that
$$I_v(z)  \ = \ (\frac{z}{2})^v\sum_{n=0}^{\infty}\frac{(z^2/4)^n}{n!\Gamma(v+n+1)}.$$
  	
\item For $\ k \in \mathbb{N}, \ $ we have that $\ I_k(z) \ = \ i_k(\frac{z}{2}, \frac{z}{2}).$

\end{itemize}

}
\end{lem}

\

\begin{thm}{\em Consider the directed manifold $\ \displaystyle (\mathbb{R}^2, \ \frac{\partial}{\partial x}, \ \frac{\partial}{\partial y}).\ $

\begin{enumerate}
\item There are no directed paths from $\ (0,0)\ $ to a point $\ (x,y) \notin \mathbb{R}_{\geq 0}^2.$
  \item  $\mathrm{vol}(\Gamma(x,0)) \ = \ \mathrm{vol}(\Gamma(0,x))  \ = \ 1, \ $ for $\ x \in \mathbb{R}_{>0}.$
  \item  For $\ (x,y) \in \mathbb{R}_{>0}^2, \ $ the moduli space $\ \Gamma(x,y)\ $ of directed paths from $\ (0,0) \ $ to $\ (x,y)\ $ has volume $$\mathrm{vol}(\Gamma(x,y))\ =  \  i_{-1}(x,y)  \ + \ 2i_0(x,y) \ + \ i_1(x,y)\ = $$
      $$\sum_{n=0}^\infty \Big( \frac{x^{n+1}y^{n}}{(n+1)!n!} \ + \ 2\frac{x^{n}y^n}{n!^2}
      \ +\ \frac{x^{n}y^{n+1}}{n!(n+1)!}\Big).$$
  \item $\mathrm{vol}(\Gamma(x,y)) \ $ is a  symmetric function in $\ x \ $ and $\ y.\ $
  \item The derivatives of the function $\ \mathrm{vol}(\Gamma) \ = \ \mathrm{vol}(\Gamma(x,y))\ $ are given by:
     $$\frac{\partial^l}{\partial x^l}\frac{\partial^m}{\partial y^m}\mathrm{vol}(\Gamma)\  =  \  i_{l-m-1}(x,y)
      \ + \ 2i_{l-m}(x,y) \ + \ i_{l-m+1}(x,y) .$$

\item We have that  $\ \ \frac{\partial}{\partial x}\frac{\partial}{\partial y}\mathrm{vol}(\Gamma)\ = \ \mathrm{vol}(\Gamma) .$

  \item Only points $\ (x,y) \in \mathbb{R}_{\geq 0}^2\ $ on the segment $\ x+y = t\ $ receive an influence from $\ (0,0) \ $ at time $\ t \geq  0. \ $ Among the points on this segment,  the highest influence from $\ (0,0) \ $ is exerted on the point $\ (\frac{t}{2}, \frac{t}{2}); \ $ the volume of the moduli space of directed paths  from $\ (0,0)\  $ along the line of maximal influences is given by
      $$\mathrm{vol}(\Gamma(t,t))\  =  \ 2\sum_{n=0}^\infty {n \choose \lfloor n/2 \rfloor }\frac{t^{n}}{n!}. $$

  \item The wave of influences $\ u(x,y, t)\ $ is given for $\  t > 0 \ $ by
  $$u(x,y,t)\  =  \ \int_{0}^{t} \mathrm{vol}(\Gamma_{(x-s,y+s-t),(x,y)}(t))ds  \ \ = \ \ 2(e^t-1).$$

\end{enumerate}

}
\end{thm}

\begin{proof}

Item 1 is clear, and item 2 simply counts the influences that arise, respectively, from the patterns $\ (1)\ $ and $\ (2). \ $
Let us show 3. Since $\ k=2, \ $ a pattern $\ (c_0,...,c_n)\ $ is determined by its initial value $\ c_0. \ $ For $\ (x,y) \in \mathbb{R}_{>0}^2\ $ we distinguish  four cases taking into account the initial value $\ c_0\ $ and the parity of $\ n$.

\begin{itemize}

\item Patterns $\ (1,2, ... , 1,2)\ $ and $\ (2,1, ... , 2,1) \ $ of length $\ 2n, \ $ for $\ n\geq 1, $  have a contribution of $$\mathrm{vol}(\Delta_{n-1}^x)\ \mathrm{vol}(\Delta_{n-1}^y) \ =  \ \frac{x^{n-1}y^{n-1}}{(n-1)!^2}$$ to the volume of the moduli space
     of directed paths.
\item The pattern $\ (1,2, ... , 1,2,1)\ $ of length $\ 2n+1, \ $ for $\ n \geq 1, \ $  have a contribution of
 $$\mathrm{vol}(\Delta_{n}^x)\ \mathrm{vol}(\Delta_{n-1}^y) \  =  \ \frac{x^{n}y^{n-1}}{n!(n-1)!}$$ to the volume of the moduli space of directed paths.
\item The pattern $\ (2,1, ... ,2,1,2)\ $ of length $\ 2n+1, \ $ for $\ n \geq 1, \ $  have a contribution of
 $$\mathrm{vol}(\Delta_{n-1}^x)\ \mathrm{vol}(\Delta_{n}^y) \  = \ \frac{x^{n-1}y^{n}}{(n-1)!n!}$$ to the volume of the moduli space of  directed paths.
\end{itemize}
 Putting together the three summands we obtain that
 $$\mathrm{vol}(\Gamma(x,y))\   = \ \sum_{n=1}^\infty\Big(2\frac{x^{n-1}y^{n-1}}{(n-1)!^2}\ + \
      \frac{x^{n}y^{n-1}}{n!(n-1)!}\ + \ \frac{x^{n-1}y^{n}}{(n-1)!n!}\Big),$$
an expression equivalent to our desired result after a change of variables. Clearly, $\ \mathrm{vol}(\Gamma(x,y))\ $ is symmetric in $\ x \ $ and $\ y, \ $ thus item 4 follows.\\

Item 5 follows from item 3 and Lemma \ref{sal}. Item 6 is a particular case of item 5.  Let us show item 7. Let $\ \mathrm{vol}_n(\Gamma(x,y))\ $ be the $n$-th coefficient in the series expansion of $\ \mathrm{vol}(\Gamma(x,y))\ $ from  item 3.
The points influenced by $\ (0,0) \ $ at time $\ t \ $  are  of the form $\ (s,t-s)\ $ with $\ 0<s<t. \ $  Thus:
$$\mathrm{vol}_n(\Gamma(s,t-s)) \ =  \  (st-s^2)^{n-1}
\Big(\frac{2}{(n-1)!^2}\ +  \ \frac{t}{(n-1)!n!}\Big).$$
Therefore
$$\frac{\partial}{\partial s}\mathrm{vol}_n(\Gamma(s,t-s)) \  =  \ (n-1)(st-s^2)^{n-2}(t-2s)
\Big(\frac{2}{(n-1)!^2}\ + \ \frac{t}{(n-1)!n!}\Big).$$
The sign of the expression above is determined by the sign of $\ (t-2s), \ $ as the other factors are positive.
Thus the volume of the moduli space of directed paths from $\ (0,0) \ $ exerted on time $\ t \ $ achieves a global maximum at the point $\  (\frac{t}{2}, \frac{t}{2}), \ $ and
we have that $$\mathrm{vol}(\Gamma(t,t))\  =  \ 2\sum_{n=0}^\infty\Big(\frac{t^{2n}}{n!^2}\ + \
\frac{t^{2n+1}}{(n+1)!n!}\Big)\  = $$
$$2\sum_{n=0}^\infty\Big({2n \choose n }\frac{t^{2n}}{(2n)!}\ + \
{2n+1 \choose n } \frac{t^{2n+1}}{(2n+1)!}\Big) \  =  \
2\sum_{n=0}^\infty {n \choose \lfloor n/2 \rfloor }\frac{t^{n}}{n!}. $$

\noindent Item 8. By translation invariance the wave of influence is independent of $ \ x,y. \ $
Thus we have that
$$u(x,y,t) \  =  \ u(0,0,t) \ =  \ \int_{0}^{t} \mathrm{vol}(\Gamma_{(-s,s-t),(0,0)}(t))ds \ =  \ \int_{0}^{t} \mathrm{vol}(\Gamma_{(0,0),(s,t-s)}(t))ds \ = $$
$$\int_{0}^{t} \Gamma(s,t-s)ds \ =  \  \sum_{n=0}^\infty\int_{0}^{t}\Big(2\frac{s^{n}(t-s)^n}{n!^2}\ + \
      \frac{s^{n+1}(t-s)^{n}}{(n+1)!n!}\ +\ \frac{s^{n}(t-s)^{n+1}}{n!(n+1)!}\Big)ds \  = $$
$$2\sum_{n=0}^\infty\int_{0}^{t}\frac{s^{n}(t-s)^n}{n!^2} ds\ + \
      2\int_{0}^{t}\frac{s^{n+1}(t-s)^{n}}{(n+1)!n!}ds\ \ = $$
$$2\sum_{n=0}^\infty \frac{t^{2n+1}}{(2n+1)!}  \ \ \ + \ \ \  2\sum_{n=0}^\infty\frac{t^{2n+2}}{(2n+2)!} \ = \ 2(\mathrm{sinh}(t) +  \mathrm{cosh}(t)-1) \ = \ 2(e^t-1). $$
\end{proof}

Next we consider the moduli spaces of directed paths on the torus $\ T^2 =  S^1\times S^1. \ $ We use coordinates $\ (x,y) \in \mathbb{R}^2\ $ representing the point $\ (e^{2\pi i x}, e^{2\pi i y}) \in T^2. \ $ Consider the vector fields on $\ T^2 \ $ given in local coordinates by
$$ \frac{\partial}{\partial x} \ \ \ \ \ \ \mbox{and} \ \ \ \ \ \  \ \frac{\partial}{\partial y}.$$
The moduli space of directed paths on the torus $\ T^2\ $ from $\ (1,1)\ $ to $\ (e^{2\pi i x},e^{2\pi i y})\ $  exerted in time $\ t>0 \ $ is denoted by
$\Gamma(e^{2\pi i x},e^{2\pi i y},t).$  Recall that $ \ D(e^{2\pi i x},e^{2\pi i y},t) \ $ is the set of one-direction paths.

\begin{thm}{\em Consider the directed manifold $\ (T^2,\frac{\partial}{\partial x},\frac{\partial}{\partial y} ). \ $

\begin{enumerate}

\item For $\ x,y \in (0,1) \ $ we have that $\ \ \mathrm{vol}(D(e^{2\pi i x},e^{2\pi i y},t)) \ = \ 0. \ $
\item For $\ x \in (0,1] \ $ we have that $$\ \ \mathrm{vol}(D(e^{2\pi i x},1,t)) \ =  \ \mathrm{vol}(D(1,e^{2\pi i x},t))  \  =  \ \sum_{m=0}^{\infty}\delta(t,x+m). \ $$
  \item  For $\ (x,y) \in (0,1)^2,\ $ the moduli space $\ \Gamma(e^{2\pi i x},e^{2\pi i y},t)\ $ of directed paths from $\ (1,1) \ $ to
  $\ (e^{2\pi i x},e^{2\pi i y})\ $ is empty unless $\ t=x+y+m\ $ for some $\ m \geq 0, \ $ and in the latter case we have that:
       $\ \mathrm{vol}(\Gamma(e^{2\pi i x},e^{2\pi i y}, x+y+m)) \ $ is given by
       $$\sum_{k+l=m}\sum_{n=0}^\infty\Big(2\frac{(x+k)^{n}(y+l)^n}{n!^2}\ + \
      (x+y+k+l)\frac{(x+k)^{n}(y+l)^{n}}{(n+1)!n!}\Big).$$
  \item $\mathrm{vol}(\Gamma(e^{2\pi i x},e^{2\pi i y}, x+y+m))\ $ is a  symmetric function in $\ x\ $ and $\ y.\ $

\end{enumerate}

}
\end{thm}

\begin{proof}
We can compute indirect influences on the torus as sums of indirect influences on the plane, indeed we have that
$$\mathrm{vol}(\Gamma(e^{2\pi i x},e^{2\pi i y}, x+y+m))\  = \  \sum_{k+l=m}\mathrm{vol}(\Gamma( x +k,y+l, x+y+m))\ \ = $$
$$\sum_{k+l=m}\sum_{n=0}^\infty\Big(2\frac{(x+k)^{n}(y+l)^n}{n!^2}\ + \
      \frac{(x+k)^{n+1}(y+l)^{n}}{(n+1)!n!}\ + \ \frac{(x+k)^{n}(y+l)^{n+1}}{n!(n+1)!}\Big). $$
\end{proof}

\subsection{Higher Dimensions}

Let us first introduce a few combinatorial notions. Given integers $\ n_1, \ldots, n_k \in \mathbb{N}_{>0}\ $ we let $\ \mathrm{Sh}_k(n_1,\dots, n_k)\ $ be the set of shuffles of $\ n_1 +\cdots + n_k\ $ cards divided into $k$ blocks of cardinalities $\ n_1, \ldots, n_k. \ $ Recall that a shuffle is a bijection $\ \alpha\ $ from the set $$[1, n_1 +\cdots + n_k] \ \simeq \ [1, n_1] \sqcup \cdots \sqcup [1, n_k]$$ to itself such that if $\ i \ < \ j \in [1, n_s], \ $ then $\ \alpha(i) \ < \ \alpha(j) \ \in\ [1, n_1 +\cdots + n_k]. \ $ When we shuffle a deck of cards the idea is to intertwine the cards in the various blocks, without distorting the order in each block. We say that a shuffle is perfect if no contiguous cards within a block remain contiguous after shuffling, i.e. a shuffle $\alpha$ is called perfect if for $\ i,\ i+1\ \in \ [1, n_s] \ $ we have that
$$\alpha(i) + 1 \ <  \  \alpha(i+1)  \  \in  \ [1, n_1 +\cdots + n_k].$$
Let $\ \mathrm{PSh}_k(n_1,\dots, n_k) \ \subseteq \ \mathrm{Sh}_k(n_1,\dots, n_k)\ $ be the set of perfect shuffles, and
$\ \mathrm{psh}_k \ $  be the corresponding exponential generating series given by
$$\mathrm{psh}_k(x_1, \ldots, x_k) \  =  \ \sum_{n_1, \ldots, n_k \in \mathbb{N}_{>0}}|\mathrm{PSh}_k(n_1,\dots, n_k)| \frac{x_1^{n_1}\cdots x_k^{n_k} }{n_1! \cdots n_k!} .$$

 A subset $\ A \ \subseteq \ [m]\ $
is called sparse if it does not contain  consecutive elements.
Let $\ S_k[m]\ $ be the set of all sparse subsets of $\ [m]\ $ of cardinality $\ k. \ $
Let $\ p(m,k)\ $ count the numerical partitions of $\ m \ $ in $\ k \ $ positive summands. \\

\begin{lem}{\em For $\ 1 \leq k < m \in \mathbb{N},$ we have that:
$$|S_k[m]|\  = \  p(m-k,k-1) \  +  \ 2p(m-k,k) \  +  \ p(m-k,k+1).$$
}
\end{lem}

\begin{proof}If $\ A \in S_k[m],\ $ then $\ |A^c|=m-k, \ $ and $\ A^c \ $ comes with a naturally ordered partition with exactly $\ k-1 \ $ blocks if $\ 1,m \in A, \ $ $k$ blocks if $\ 1 \ $ or $\ m \ $ (but not both) belong to $\ A, \ $ and $\ k+1 \ $ blocks if $\ 1,m \notin A.\ $ The cardinalities of the blocks of $\ A^c \ $ provides the various kinds of numerical partitions needed to complete our result.
\end{proof}

\begin{lem}{\em For $\ n_1,\dots, n_k \in \mathbb{N}_{>0},\ $ then $\ |\mathrm{PSh}_k(n_1,\dots, n_k)| \ $ counts number of ordered partitions of
$\ n_1\ + \ \dots \ + \ n_k\ $ with sparse blocks of cardinalities $\ n_1,\dots, n_k.$
}
\end{lem}

\begin{proof}A perfect shuffle in $\ \mathrm{PSh}_k(n_1,\dots, n_k)\ $ is determined by its image on each of the blocks $\ [1,n_s], \ $ which must be a sparse subsets.
\end{proof}

Let us point out the relation between patterns and perfect shuffles. Consider the map
$$|\ |: C(n,k)\ \longrightarrow \ \mathbb{N}^k ,$$ sending a pattern $\ c\in C(n,k)\ $ to its content multi-set given by the  sequence $\ |c| \in \mathbb{N}^k\ $ such that $\ |c|_i \ = \ |c^{-1}(i)|. \ $ The support of a pattern $\ c \ $ is the set $\ s(c) \ \subseteq \ [k] \ $ with $\ i\in s(c)\ $ if and only if $\ |c|_i\ \neq \ 0.$

\begin{lem}{\em Fix a vector $\ (n_1,\dots, n_k) \in \mathbb{N}_{>0}^k. \ $ We have that:
$$\bigg|\{ c\in C(n,k) \ | \ |c|=(n_1,...,n_k) \}\bigg| \  =  \ \big|\mathrm{PSh}_k(n_1,\dots, n_k)\big|.$$
}
\end{lem}

\begin{proof}The vector $\ (n_1,\dots, n_k)\ $ gives us the content multi-set of $\ c, \ $ a shuffle on it gives us in addition the order of the vector $\ c. \ $ The perfect condition on shuffles is equivalent to the conditions $\ c(i)\ \neq \ c(i+1) \ $ on patterns.
\end{proof}

Consider the directed manifold $\ (\mathbb{R}^k, \ \frac{\partial}{\partial x_1}, \dots , \frac{\partial}{\partial x_k}). \ $ The moduli space of directed paths from $\ (0,\ldots,0)\ $ to $\ (x_1, \ldots, x_k)\ $ is denoted by $\ \Gamma(x_1,\ldots, x_k). \ $ Such paths  can only happen at time $\ t\ = \ x_1+ \cdots + x_k. \ $
\begin{thm}{\em Consider the directed manifold $\ (\mathbb{R}^k, \ \frac{\partial}{\partial x_1}, \dots , \frac{\partial}{\partial x_k}). \ $

\begin{enumerate}
\item There are no directed paths from $(0,\ldots,0) \ $ to any point $\ (x_1, \ldots, x_k) \notin \mathbb{R}_{\geq 0}^k.$
  \item  $\mathrm{vol}(D(0, \ldots, 0, \underset{i\uparrow}{x}, 0, \ldots, 0) \ = \  1, \  $ for $ \ x \in \mathbb{R}_{\geq 0}\ \ $ and $\ \ i \in [k].$
  \item  For $\ (x_1, \ldots, x_k) \in \mathbb{R}_{\geq 0}^k, \ $  with at least two positive entries, the moduli space $\ \Gamma(x_1, \ldots, x_k)\ $ of directed paths from $\ (0,\ldots,0)\ $ to $\ (x_1, \ldots, x_k)\ $ has volume $$\mathrm{vol}(\Gamma(x_1, \ldots, x_k))\  =  \  \underset{\underset{|A| \geq 2} {A \subseteq [k]}}{\sum}\ \frac{\partial^{|A|}}{\partial x_A}\mathrm{psh}_{|A|} (x_A).$$
  \item $\mathrm{vol}(\Gamma(x_1, \ldots, x_k))\ $ is a  symmetric function in the variables $\ x_1, \ldots, x_k.$

\end{enumerate}
}
\end{thm}

\begin{proof}Properties 1 and 2 are clear, let us prove 3. Recall that
$$\mathrm{vol}(\Gamma (x_1, \ldots, x_k)) \  =  \ \sum_{n=1}^{\infty}\ \sum_{c\in C(n,k)}\mathrm{vol}(\Gamma^c (x_1, \ldots, x_k)) ,$$ where the volume of the moduli space of directed paths with a fix pattern $\ c \in C(n,k) \ $ is given by
$$\mathrm{vol}(\Gamma^c (x_1, \ldots, x_k))  \ =  \
\prod_{j \in s(c)}\frac{x_j^{|c|_j - 1}}{(|c|_j -1)!}.$$ Thus a pattern $\ c \in C(n,k) \ $  with support $\ s(c)\ = \ A \  \subseteq \ [k], \ $ with $\ |A|\geq 2, \ $   contributes to the monomial
 $$\frac{x_1^{n_1}\cdots x_k^{n_k}}{n_1! \cdots n_k!},$$ if and only if $\ \ |c|_i\ = \ n_i+1\ \ $ for $\ i\in A,\ \ $  and $\ \  n_i=0 \ $ for $\ i \notin A.\ \ $
  Therefore the total contribution of the patterns with support $\ A \ $ to this monomial is given by
 $$|\mathrm{PSh}_{|A|}(n_A +1)|\prod_{j \in A}\frac{x_j^{n_j}}{n_j!},$$
 where $\ n_A\ $ is the vector obtained from the tuple $\ (n_1,...,n_k)\ $ by erasing the zero entries, and  $\ n_A +1 \ $ is the vector
 obtain from $\ n_A\ $ by adding $\ 1\ $ to each entry.\\

Summing  over the $\ n_j, \ $ and setting $\ x_A  =  (x_j)_{j \in A},\ $  we obtain that the total contribution
of the patterns with support $\ A \ $ to the volume of the moduli space of direct
ed paths is given by
$$\sum_{n_j \in \mathbb{N}; \ j\in A}| \mathrm{PSh}_{|A|}(n_A +1)|\prod_{j \in A}\frac{x_j^{n_j}}{n_j!}\  = \ \frac{\partial^{|A|}}{\partial x_A}\mathrm{psh}_{|A|} (x_A).$$
Adding over all possible supports $\ A \subseteq [k], \ $ with $\ |A| \geq 2, \ $ we obtain the desired result.\\

\noindent 4. For a permutation $\ \sigma \in S_n \ $ we have that
$\ \mathrm{vol}(\Gamma(x_{\sigma(1)}, \ldots, x_{\sigma(k)})\ $ is given by  $$ \ \underset{\underset{|A| \geq 2} {A \subseteq [k]}}{\sum}\frac{\partial^{|\sigma A|}}{\partial x_{\sigma A}}\mathrm{psh}_{|\sigma A|} (x_{\sigma A}) \ = \ \underset{\underset{|A| \geq 2} {A \subseteq [k]}}{\sum}\frac{\partial^{| A|}}{\partial x_{ A}}\mathrm{psh}_{| A|} (x_{ A}) \  =  \ \mathrm{vol}(\Gamma (x_1, \ldots, x_k)) .$$

\end{proof}

Next we consider direct
ed paths on the $k$-dimensional torus $\ T^k \ = \ S^1\times \cdots\times S^1. \ $ We use coordinates $\ (x_1, \ldots, x_k) \in \mathbb{R}^k\ $ representing the point $\ (e^{2\pi i x_1}, \ldots, e^{2\pi i x_k}) \in T^k. \ $ Consider the constant vector fields on $\ T^k \ $ given in local coordinates by
$$ \frac{\partial}{\partial x_1}, \ \ \cdots \ \  ,\frac{\partial}{\partial x_k}.$$
The moduli space of directed paths on $\ T^k\ $ from $\ (1,...,1)\ $ to $\ (e^{2\pi i x_1}, \ldots, e^{2\pi i x_k})\ $  exerted in time $\ t>0 \ $ is denoted by $\ \Gamma(e^{2\pi i x_1},\ldots, e^{2\pi i x_k},t).$ Recall that the set of one-direction paths  is denoted by $\ D(e^{2\pi i x_1},\ldots, e^{2\pi i x_k},t).$

\begin{thm}{\em Consider the directed manifold $\ (T^k,  \frac{\partial}{\partial x_1}, \ \cdots \  ,\frac{\partial}{\partial x_k}).\ $ \\

\noindent 1. For $\ x_1, \ldots, x_k \ \in \ (0,1], \ $ with at least two entries in $\ (0,1), \ $  we have that
$$ \mathrm{vol}(D(e^{2\pi i x_1},\ldots, e^{2\pi i x_k},t)) \ =  \ 0.$$

\noindent 2. For $\ x \ \in \ (0,1]\ $ we have that:
$$\mathrm{vol}(D(1, \dots, \underset{i\uparrow}{e^{2\pi i x}} , \ldots, 1,t)) \ =  \  \sum_{m=0}^{\infty}\delta(t,x_i + m). \ $$

\noindent 3. For $\ x_1, \ldots, x_k \ \in \ (0,1],\ $ with at least two entries in $\ (0,1), \ $ the moduli space $\ \Gamma(e^{2\pi i x_1},\ldots, e^{2\pi i x_k},t)\ $ of directed paths from $\ (1, \ldots,1)\ $ to  $\ (e^{2\pi i x_1},\ldots, e^{2\pi i x_k})\ $ is empty unless $\ t \ = \ x_1+ \cdots + x_k+m\ $ for some $\ m \geq 0, \ $ and in the latter case we have that:
       $$\mathrm{vol}(\Gamma(e^{2\pi i x_1},\ldots, e^{2\pi i x_k}, x_1+ \cdots + x_k+m)) \ \ = \ \
       \sum_{m_1+\ldots + m_k=m}\ \underset{\underset{|A| \geq 2} {A \subseteq [d]}}{\sum}\ \frac{\partial^{|A|}}{\partial x_A}\mathrm{psh}_{|A|} (x_A \ + \ m_A).$$
\noindent 4. $\mathrm{vol}(\Gamma(e^{2\pi i x_1},\ldots, e^{2\pi i x_k},  x_1+ \cdots + x_k+m))\ \ $ is a  symmetric function on $\ x_1, \ldots, x_k.$

}
\end{thm}

\section{Quantum Indirect Influences}\label{qii}

In this closing section we briefly describe how to extend the theory of indirect influences to the quantum settings.
We first consider indirect influences on Poisson manifolds \cite{arnold} from two different viewpoints.\\

Let $\ (M,\{\ , \ \} )\ $ be a Poisson manifold and $\ f_1,...,f_k:M \longrightarrow \mathbb{R}\ $ be $\ k \ $ smooth functions on $\ M. \ $
We obtain the directed manifold
$$(M, \ \{f_1, \ \}, \ \dots \ ,\{f_k, \ \}), $$
where $\ \{f_j, \ \}, \  $  the Hamiltonian vector field on $\ M \ $ generated by $\ f_j, \ $ is given in local coordinates by
$$\{f_j, \ \} \ =  \ \sum_{kl}\{x_k, x_l \}\frac{\partial f_j}{\partial x_k}\frac{\partial}{\partial x_l}.$$

Next let $\ C^{\infty}(M) \ $ be the infinite dimensional vector space of smooth functions on $\ M .\ $ We obtain the infinite dimensional directed manifold
$$(C^{\infty}(M),  \ \{f_1, \ \}, \ \dots \ ,\{f_k, \ \} )$$  where now we regard $\  \{f_j, \ \} \ $ as the vector field on  $\ C^{\infty}(M) \ $ assigning
to $f \in C^{\infty}(M) \ $ the vector $$  \{f_j, f \}\ \in \ T_fC^{\infty}(M) \ = \ C^{\infty}(M). $$
Given functions $\ f, g \ \in \ C^{\infty}(M) \ $  and a pattern $\ c \in D(n,k)\ $ the moduli space of directed paths
from $\ f \ $ to $\ g \ $ exerted in time $\ t>0 \ $ is given by
$$\Gamma_{f,g}^{c}(t) \  =  \ \{\ s \in \Delta_n^t  \ \ | \ \ \varphi_c(f,s)\ = \ g \ \} .$$
The flow generated by $\ \{f_j, \ \}\ $ on $\ M ,\ $ and the flow generated by $\ \{f_j, \ \}\ $ on $\ C^{\infty}(M)\ $ (allow us to use the same notation for vector fields in different spaces) are related by the identity
$$\varphi_j(f,s)(x) \ =  \ f(\varphi_j(x,s)) .$$
Since the vector fields  $\ \{f_j, \ \}\ $ are linear operators on $\ C^{\infty}(M),\ $ the flows generated by them -- assuming suitable
convergency properties -- can be written as
$$\varphi_j(f,s) \  =  \ e^{\{f_j, \ \}s}f.$$
Expanding the exponentials functions the iterated flow $\ \varphi_c(f,s) \ $ can be written as
$$\varphi_c(f,s) \  =  \ \sum_{k_0,...,k_n \in \mathbb{N}}\{f_{c(n)},\dots,f_{c(0)}, f\}_{k_0,\dots, k_n}\frac{s_0^{k_0}\cdots s_n^{k_n}}{k_0!\cdots k_n!},$$
where the symbol $\ \{g_{n},\dots,g_{1}, f\}_{k_1,\dots, k_n}  \ $ is defined recursively as follows:
$$\{g_1,f\}_0 \  =  \ f, \ \ \ \ \ \ \ \ \ \ \  \{g, f\}_{k+1} \  =  \ \{g, \{g, f\}_k \},$$
$$\{g_{n},\dots,g_{1}, f\}_{k_1,\dots,k_{n-1}, k_n} \  = \   \{g_n, \ \{g_{n-1}\dots,g_{1}, f\}_{k_1,\dots,k_{n-1}}\}_{k_n}.$$

\

From this viewpoint it is clear how to extend the theory of indirect influences to the quantum context \cite{glimjaffe,SimonFunctional}.
 Let $\ \mathcal{H} \ $ be a Hilbert space and $\ A_1,...,A_k\ $ be bounded Hermitian operators on $\ \mathcal{H}. \ $ \\

In the Heisenberg picture we consider the (possibly infinite dimensional) directed manifold
$$(\mathcal{B(H)},\  \frac{i}{\hbar}[A_1, \ ], \ \dots \ , \frac{i}{\hbar}[A_k, \ ]) $$ where $\ \mathcal{B(H)} \ $ is the algebra of bounded
 operators on $\ \mathcal{H}, \ $ $\ [\ , \ ] \ $ is the commutator of bounded operators, and $\ \frac{i}{\hbar}[A_1, \ ] \ $ is
 regarded as the vector field on $\ \mathcal{B(H)} \ $ assigning  to $\ B \in \mathcal{B(H)}\ $ the vector $$\ \frac{i}{\hbar}[A_1, B ] \ \in \ T_B\mathcal{B(H)}\ = \ \mathcal{B(H)}. $$
Given operators $\ B, C \ \in \ \mathcal{B(H)} \ $  and a pattern $\ c \in D(n,k),\ $ the moduli space of directed paths
from $\ B \ $ to $\ C \ $ exerted in time $\ t>0 \ $ is given by
$$\Gamma_{B,C}^{c}(t) \  =  \ \{\ s \in \Delta_n^t  \ \ | \ \ \varphi_c(B,s)\ = \ C \ \} ,$$
where the iterated flow $\ \varphi_c(B,s) \ $ is given by
$$\varphi_c(B,s) \ \ = \ \ e^{\frac{i}{\hbar}A_{c(n)}s_n}\cdots e^{\frac{i}{\hbar}A_{c(1)}s_1}e^{\frac{i}{\hbar}A_{c(0)}s_0}B
e^{-\frac{i}{\hbar}A_{c(0)}s_0}e^{-\frac{i}{\hbar}A_{c(1)s_1}}e^{-\frac{i}{\hbar}A_{c(n)}s_n} \ \ = $$
$$\sum_{k_0,...,k_n \in \mathbb{N}}(\frac{i}{\hbar})^{k_0+ \cdots + k_n}\big[A_{c(n)},\dots,A_{c(0)}, B\big]_{k_0,\dots, k_n}\frac{s_0^{k_0}\cdots s_n^{k_n}}{k_0!\cdots k_n!}$$
where the symbols $$ \big[A_{c(n)},\dots,A_{c(0)}, B\big]_{k_0,\dots, k_n}$$ are defined as in the Poisson case replacing brackets $\ \{ \ , \ \}\ $ by commutators $\ [\ , \ ].$
Clearly, we can apply this constructions in the context of deformation quantization as well \cite{kont}.\\

In the Schr$\ddot{\mbox{o}}$dinger picture we consider the (possibly infinite dimensional) directed manifold
$$(\mathcal{H},\  -\frac{i}{\hbar}A_1, \ \dots \ , -\frac{i}{\hbar}A_k), $$
where $\ -\frac{i}{\hbar}A_j\ $ is regarded as the vector field assigning to $\ v \in \mathcal{H} \ $ the vector
$$-\frac{i}{\hbar}A_j(v) \ \in \ T_v\mathcal{H }\ = \ \mathcal{H}. $$
Given $\ v, w \ \in \ \mathcal{H} \ $  and a pattern $\ c \in D(n,k),\ $ the moduli space of directed paths
from $\ v \ $ to $\ w \ $ exerted in time $\ t>0 \ $ is given by
$$\Gamma_{v,w}^{c}(t) \  =  \ \{\ s \in \Delta_n^t  \ \ | \ \ \varphi_c(v,s)\ = \ w \ \} .$$
The iterated flow $\ \varphi_c(v,s) \ $ is given by
$$\varphi_c(v,s) \  =  \ \sum_{k_0,...,k_n \in \mathbb{N}}(-\frac{i}{\hbar})^{k_0+ \cdots + k_n}\big(A_{c(n)}^{k_n},\dots,A_{c(0)}^{k_0}v\big) \frac{s_0^{k_0}\cdots s_n^{k_n}}{k_0!\cdots k_n!}.$$

\subsection*{Acknowledgement}
Our thanks to Tom Koornwinder whose comments and suggestions help us to make substantial improvements on an early version of this work.

\noindent lnrdcano@gmail.com \\
\noindent Departamento de Matem\'aticas, Universidad Sergio Arboleda, Bogot\'a, Colombia\\

\noindent ragadiaz@gmail.com\\
\noindent Departamento de Matem\'aticas, Pontificia Universidad Javeriana, Bogot\'a, Colombia

\end{document}